\documentclass[11pt]{article}
\usepackage[left=.8in, right=.8in, top=.9in, bottom=.9in]{geometry}
\usepackage{amssymb}
\usepackage{amsmath} 
\usepackage{amsthm}
\usepackage{stackengine}
\usepackage{verbatim}
\usepackage{mathtools}
\usepackage{enumerate}
\usepackage{hyperref}
\hypersetup{
	colorlinks=true,
	citecolor=blue,
	linkcolor=blue,
	urlcolor=blue
}
\usepackage{cleveref}

\DeclareMathOperator{\U}{U}
\DeclareMathOperator{\T}{T}

\DeclareMathOperator{\Tr}{Tr}

\DeclareMathOperator{\rank}{rank}

\DeclareMathOperator{\herm}{Herm}
\DeclareMathOperator{\pos}{Pos}
\DeclareMathOperator{\proj}{Proj}
\DeclareMathOperator{\den}{Den}
\DeclareMathOperator{\chan}{Chan}
\DeclareMathOperator{\quantop}{QuantOp}
\let\L\relax
\DeclareMathOperator{\L}{L}

\newcommand{\bra}[1]{\left\langle#1\right|}
\newcommand{\ket}[1]{\left|#1\right\rangle}
\newcommand{\braket}[2]{\left\langle#1\vert#2\right\rangle}

\def\Balancedhookarrowleft{\hbox{$\ensurestackMath{\stackanchor[.42pt]{%
				\scriptscriptstyle-\mkern-10mu-}{\scriptscriptstyle\leftarrow}}\mkern-6mu%
		\raisebox{1.82pt}{$\scriptscriptstyle\supset$}$}}

\def\Balancedhookarrowright{\reflectbox{\Balancedhookarrowleft}}

\DeclareMathAlphabet{\mathbbm}{U}{bbold}{m}{n}

\newtheorem{theorem}{Theorem}
\newtheorem{corollary}[theorem]{Corollary}
\newtheorem{lemma}[theorem]{Lemma}
\newtheorem{proposition}[theorem]{Proposition}

\theoremstyle{remark}
\newtheorem*{remark}{Remark}

\theoremstyle{definition}
\newtheorem{definition2}[theorem]{Definition}

\theoremstyle{remark}
\newtheorem{openProb}{Open Problem}

\title{Lower Bounds on the Running Time of Two-Way Quantum Finite Automata and Sublogarithmic-Space Quantum Turing Machines}
\author{Zachary Remscrim \\ Department of Computer Science, The University of Chicago \\ remscrim@uchicago.edu}
\date{}
\begin{document}

\maketitle

\begin{abstract}
	
	The two-way finite automaton with quantum and classical states (2QCFA), defined by Ambainis and Watrous, is a model of quantum computation whose quantum part is extremely limited; however, as they showed, 2QCFA are surprisingly powerful: a 2QCFA with only a single-qubit can recognize the language $L_{pal}=\{w \in \{a,b\}^*:w \text{ is a palindrome}\}$ with bounded error in expected time $2^{O(n)}$. 
	
	We prove that their result cannot be improved upon: a 2QCFA (of any size) cannot recognize $L_{pal}$ with bounded error in expected time $2^{o(n)}$. This is the first example of a language that can be recognized with bounded error by a 2QCFA in exponential time but not in subexponential time. Moreover, we prove that a quantum Turing machine (QTM) running in space $o(\log n)$ and expected time $2^{n^{1-\Omega(1)}}$ cannot recognize $L_{pal}$ with bounded error; again, this is the first lower bound of its kind. Far more generally, we establish a lower bound on the running time of any 2QCFA or $o(\log n)$-space QTM that recognizes any language $L$ in terms of a natural ``hardness measure'' of $L$. This allows us to exhibit a large family of languages for which we have asymptotically matching lower and upper bounds on the running time of any such 2QCFA or QTM recognizer.
	
\end{abstract}

\section{Introduction}\label{sec:intro}

Quantum algorithms, such as Shor's quantum polynomial time integer factorization algorithm \cite{shor1994algorithms}, Grover's algorithm for unstructured search \cite{grover1996fast}, and the linear system solver of Harrow, Hassidim, and Lloyd \cite{harrow2009quantum}, provide examples of natural problems on which quantum computers seem to have an advantage over their classical counterparts. However, these algorithms are designed to be run on a quantum computer that has the full power of a quantum Turing machine, whereas current experimental quantum computers only possess a rather limited quantum part. In particular, current state-of-the-art quantum computers have a very small amount of quantum memory. For example, Google's ``Sycamor'' quantum computer, used in their famous recent quantum supremacy experiment \cite{arute2019quantum}, operates on only $53$ qubits.

In this paper, we study the power quantum computers that have only a small amount of memory. We begin by considering two-way finite automata with quantum and classical states (2QCFA), originally defined by Ambainis and Watrous \cite{ambainis2002two}. Informally, a 2QCFA is a two-way deterministic finite automaton (2DFA) that has been augmented by a quantum register of constant size. 2QCFA are surprisingly powerful, as originally demonstrated by Ambainis and Watrous, who showed that a 2QCFA, with only a single-qubit quantum register, can recognize, with bounded error, the language $L_{eq}=\{a^m b^m:m \in \mathbb{N}\}$ in expected time $O(n^4)$ and the language $L_{pal}=\{w \in \{a,b\}^*:w \text{ is a palindrome}\}$ in expected time $2^{O(n)}$. In a recent paper \cite{remscrim2019power}, we presented further evidence of the power of few qubits by showing that 2QCFA are capable of recognizing many group word problems with bounded error. 

It is known that 2QCFA are more powerful than 2DFA and two-way probabilistic finite automata (2PFA). A 2DFA can only recognize regular languages \cite{rabin1959finite}. A 2PFA can recognize some nonregular languages with bounded error, given sufficient running time: in particular, a 2PFA can recognize $L_{eq}$ with bounded error in expected time $2^{O(n)}$ \cite{freivalds1981probabilistic}. However, a 2PFA cannot recognize $L_{eq}$ with bounded error in expected time $2^{o(n)}$, by a result of Greenberg and Weiss \cite{greenberg1986lower}; moreover, a 2PFA cannot recognize $L_{pal}$ with bounded error in any time bound \cite{dwork1992finite}. More generally, the landmark result of Dwork and Stockmeyer \cite{dwork1990time} showed that a 2PFA cannot recognize any nonregular language in expected time $2^{n^{o(1)}}$. In order to prove this statement, they defined a particular ``hardness measure'' $D_L:\mathbb{N} \rightarrow \mathbb{N}$ of a language $L$. They showed that, if a 2PFA recognizes some language $L$ with bounded error in expected time at most $T(n)$ on all inputs of length at most $n$, then there is a positive real number $a$ (that depends only on the number of states of the 2PFA), such that $T(n)=\Omega\left(2^{D_L(n)^a}\right)$ \cite[Lemma 4.3]{dwork1990time}; we will refer to this statement as the ``Dwork-Stockmeyer lemma.''  

Very little was known about the limitations of 2QCFA. Are there any languages that a single-qubit 2QCFA can recognize with bounded error in expected exponential time but not in expected subexponential time? In particular, is it possible for a single-qubit 2QCFA to recognize $L_{pal}$ in subexponential time, or perhaps even in polynomial time? More generally, are there any languages that a 2QCFA (that is allowed to have a quantum register of any constant size) can recognize with bounded error in exponential time but not in subexponential time? These natural questions, to our knowledge, were all open (see, for instance, \cite{ambainis2002two,ambainis2015automata,yakaryilmaz2010succinctness} for previous discussions of these questions). 

In this paper, we answer these and other related questions. We first prove an analogue of the Dwork-Stockmeyer lemma for 2QCFA. 

\begin{theorem}\label{thm:intro:2qfaDworkstockmeyer}
	If a 2QCFA recognizes some language $L$ with bounded error in expected time at most $T(n)$ on all inputs of length at most $n$, then there a positive real number $a$ (that depends only on the number of states of the 2QCFA), such that $T(n)=\Omega\left(D_L(n)^a\right)$.
\end{theorem}

This immediately implies that the result of Ambainis and Watrous \cite{ambainis2002two} cannot be improved.

\begin{corollary}\label{cor:intro:2qfaPal}
	2QCFA (of any size) cannot recognize $L_{pal}$ with bounded error in time $2^{o(n)}$.
\end{corollary}

One of the key tools used in our proof is a quantum version of Hennie's \cite{hennie1965one} notion of a crossing sequence, which may be of independent interest. Crossing sequences played an important role in the aforementioned 2PFA results of Dwork and Stockmeyer \cite{dwork1990time} and of Greenberg and Weiss \cite{greenberg1986lower}. We note that, while our lower bound on the running time of a 2QCFA is exponentially weaker than the lower bound on the running time of a 2PFA provided by the Dwork-Stockmeyer lemma, both lower bounds are in fact (asymptotically) tight; the exponential difference provides yet another example of a situation in which quantum computers have an exponential advantage over their classical counterparts. We also establish a lower bound on the expected running time of a 2QCFA recognizer of $L$ in terms of the one-way deterministic communication complexity of testing membership in $L$. 

We then generalize our results to prove a lower bound on the expected running time $T(n)$ of a quantum Turing machine (QTM) that uses sublogarithmic space (i.e., $o(\log n)$ space) and recognizes a language $L$ with bounded error, where this lower bound is also in terms of $D_L(n)$. In particular, we show that $L_{pal}$ cannot be recognized with bounded error by a QTM that uses sublogarithmic space and runs in expected time $2^{n^{1-\Omega(1)}}$. This result is particularly intriguing, as $L_{pal}$ can be recognized by a \textit{deterministic} TM in $O(\log n)$ space (and, trivially, polynomial time); therefore, $L_{pal}$ provides an example of a natural problem for which polynomial time \textit{quantum} TMs have no (asymptotic) advantage over polynomial time \textit{deterministic} TMs in terms of the needed amount of space.

Furthermore, we show that the class of languages recognizable with bounded error by a 2QCFA in expected polynomial time is contained in $\mathsf{L/poly}$. This result, which shows that the class of languages recognizable by a particular quantum model is contained in the class of languages recognizable by a particular classical model, is a type of \textit{dequantization} result. It is (qualitatively) similar to the Adleman-type \cite{adleman1978two} \textit{derandomization} result $\mathsf{BPL} \subseteq \mathsf{L/poly}$, where $\mathsf{BPL}$ denotes the class of languages recognizable with bounded error by a probabilistic Turing machine (PTM) that uses $O(\log n)$ space and runs in expected polynomial time. The only previous dequantization result was of a very different type: the class of languages recognizable by a 2QCFA, or more generally a QTM that uses $O(\log n)$ space, with algebraic number transition amplitudes (even with unbounded error and with no time bound), is contained in $\mathsf{DSPACE}(O(\log^2 n))$ \cite{watrous2003complexity}. This dequantization result is analogous to the derandomization result: the class of languages recognizable by a PTM that uses $O(\log n)$ space (even with unbounded error and with no time bound), is contained in $\mathsf{DSPACE}(O(\log^2 n))$ \cite{borodin1983parallel}.

We also investigate which group word problems can be recognized by 2QCFA or QTMs with particular resource bounds. Informally, the word problem of a finitely generated group is the problem of determining if the product of a sequence of elements of that group is equal to the identity element. There is a deep connection between the algebraic properties of a finitely generated group $G$ and the complexity of its word problem $W_G$, as has been demonstrated by many famous results; for example, $W_G \in \mathsf{REG} \Leftrightarrow G$ is finite \cite{anisimov1971group}, $W_G \in \mathsf{CFL} \Leftrightarrow G$ is virtually free \cite{muller1983groups,dunwoody1985accessibility}, $W_G \in \mathsf{NP} \Leftrightarrow G$ is a subgroup of a finitely presented group with polynomial Dehn function \cite{birget2002isoperimetric}. We have recently shown that if $G$ is virtually abelian, then $W_G$ may be recognized with bounded error by a single-qubit 2QCFA in polynomial time, and that, for any group $G$ in a certain broad class of groups of exponential growth, $W_G$ may be recognized with bounded error by a 2QCFA in time $2^{O(n)}$ \cite{remscrim2019power}. 

We now show that, if $G$ has exponential growth, then $W_G$ cannot be recognized by a 2QCFA with bounded error in time $2^{o(n)}$, thereby providing a broad and natural class of languages that may be recognized by a 2QCFA in time $2^{O(n)}$ but not $2^{o(n)}$. We also show that, if $W_G$ is recognizable by a 2QCFA with bounded error in expected polynomial time, then $G$ must be virtually nilpotent (i.e., $G$ must have polynomial growth), thereby obtaining progress towards an exact classification of those word problems recognizable by a 2QCFA in polynomial time. Furthermore, we show analogous results for sublogarithmic-space QTMs.

The remainder of this paper is organized as follows. In \Cref{sec:prelim}, we briefly recall the fundamentals of quantum computation and the definition of 2QCFA. In \Cref{sec:crossingSequences}, we develop our notion of a quantum crossing sequence. The Dwork-Stockmeyer hardness measure $D_L$ of a language $L$, as well as several other related hardness measures of $L$, play a key role in our lower bounds; we recall the definitions of these hardness measures in \Cref{sec:runningTimeLowerBound:nonregularity}. Then, in Section~\ref{sec:runningTimeLowerBound:dworkStockmeyer}, using our notion of a quantum crossing sequence, we prove an analogue of the Dwork-Stockmeyer lemma for 2QCFA. Using this lemma, in \Cref{sec:runningTimeLowerBound:2qcfaCompClasses}, we establish various lower bounds on the expected running time of 2QCFA for particular languages and prove certain complexity class separations and inclusions. In \Cref{sec:sublogspaceQTM}, we establish lower bounds on the expected running time of sublogarithmic-space QTMs. In \Cref{sec:ComplexityWordProb}, we study group word problems and establish lower bounds on the expected running time of 2QCFA and sublogarithmic-space QTMs that recognize certain word problems. Finally, in \Cref{sec:discussion}, we discuss some open problems related to our work.

\section{Preliminaries}\label{sec:prelim}

\subsection{Quantum Computation}\label{sec:prelim:quant}

In this section, we briefly recall the fundamentals of quantum computation needed in this paper (see, for instance, \cite{watrous2018theory,nielsen2002quantum} for a more detailed presentation of the material in this section). We begin by establishing some notation. Let $V$ denote a finite-dimensional complex Hilbert space with inner product $\langle \cdot,\cdot \rangle : V \times V \rightarrow \mathbb{C}$. We use the standard Dirac bra-ket notation throughout this paper. We denote elements of $V$ by \textit{kets}: $\ket{\psi}$, $\ket{\varphi}$, $\ket{q}$, etc. For the \textit{ket} $\ket{\psi} \in V$, we define the corresponding \textit{bra} $\bra{\psi} \in V^*$ to be the linear functional on $V$ given by $\langle \ket{\psi}, \cdot \rangle :V \rightarrow \mathbb{C}$. We write $\braket{\psi}{\varphi}$ to denote $\langle \ket{\psi}, \ket{\varphi} \rangle$. Let $\L(V)$ denote the $\mathbb{C}$-vector space consisting of all $\mathbb{C}$-linear maps of the form $A:V \rightarrow V$. For $\ket{\psi},\ket{\varphi} \in V$, we define $\ket{\psi}\bra{\varphi} \in \L(V)$ in the natural way: for $\ket{\rho} \in V$, $\ket{\psi}\bra{\varphi}(\ket{\rho})=\ket{\psi}\braket{\varphi}{\rho}=\braket{\varphi}{\rho}\ket{\psi}$. Let $\mathbbm{1}_V \in \L(V)$ denote the identity operator on $V$ and let $\mathbbm{0}_V \in \L(V)$ denote the zero operator on $V$. For $A \in \L(V)$, we define $A^{\dagger} \in \L(V)$, the \textit{Hermitian transpose} of $A$, to be the unique element of $\L(V)$ such that $\langle A\ket{\psi_1},\ket{\psi_2} \rangle =\langle \ket{\psi_1},A^{\dagger} \ket{\psi_2}$, $\forall \ket{\psi_1},\ket{\psi_2} \in V$. Let $\herm(V)=\{A \in \L(V):A=A^{\dagger}\}$, $\pos(V)=\{A^{\dagger}A:A \in \L(V)\}$, $\proj(V)=\{A \in \pos(V):A^2=A\}$, $\U(V)=\{A \in \L(V):AA^{\dagger}=\mathbbm{1}_V\}$, and $\den(V)=\{A \in \pos(V):\Tr(A)=1\}$ denote, respectively, the set of Hermitian, positive semi-definite, projection, unitary, and density operators on $V$.

A \textit{quantum register} is specified by a finite set of \textit{quantum basis states} $Q=\{q_0,\ldots,q_{k-1}\}$. Corresponding to these $k$ quantum basis states is an orthonormal basis $\{\ket{q_0},\ldots,\ket{q_{k-1}}\}$ of the finite-dimensional complex Hilbert space $\mathbb{C}^Q \cong \mathbb{C}^k$. The quantum register stores a \textit{superposition} $\ket{\psi}=\sum_q \alpha_q \ket{q} \in \mathbb{C}^Q$, where each $\alpha_q \in \mathbb{C}$ and $\sum_q \lvert \alpha_q \rvert^2=1$; in other words, a superposition $\ket{\psi}$ is simply an element of $\mathbb{C}^Q$ of norm $1$. 

Following the original definition of Ambainis and Watrous \cite{ambainis2002two}, a 2QCFA may only interact with its quantum register in two ways: by applying a \textit{unitary transformation} or performing a \textit{quantum measurement}. If the quantum register is currently in the superposition $\ket{\psi} \in \mathbb{C}^Q$, then after applying the unitary transformation $T \in \U(\mathbb{C}^Q)$, the quantum register will be in the superposition $T \ket{\psi}$. A \textit{von Neumann measurement} is specified by some $P_1,\ldots,P_l \in \proj(\mathbb{C}^Q)$, such that $P_i P_j =\mathbbm{0}_{\mathbb{C}^Q}$, $\forall i,j$ with $i \neq j$, and $\sum_j P_j=\mathbbm{1}_{\mathbb{C}^Q}$. Quantum measurement is a probabilistic process where, if the quantum register is in the superposition $\ket{\psi}$, then the \textit{result} of the measurement has the value $r \in \{1,\ldots,l\}$ with probability $\lVert P_r \ket{\psi} \rVert^2$; if the result is $r$, then the quantum register collapses to the superposition $\frac{1}{\lVert P_r \ket{\psi} \rVert} P_r \ket{\psi}$. We emphasize that quantum measurement changes the state of the quantum register. 

An \textit{ensemble of pure states} of the quantum register is a set $\{(p_i,\ket{\psi_i}):i \in I \}$, for some index set $I$, where $p_i \in [0,1]$ denotes the probability of the quantum register being in the superposition $\ket{\psi_i}$, and $\sum_i p_i=1$. This ensemble corresponds to the density operator $A=\sum_i p_i \ket{\psi_i}\bra{\psi_i} \in \den(\mathbb{C}^Q)$. Of course, many distinct ensembles correspond to the density operator $A$; however, all ensembles that correspond to a particular density operator will behave the same, for our purposes (see, for instance, \cite[Section 2.4]{nielsen2002quantum} for a detailed discussion of this phenomenon, and of the following claims). That is to say, for any ensemble described by a density operator $A \in \den(\mathbb{C}^Q)$, applying the transformation $T \in \U(\mathbb{C}^Q)$ produces an ensemble described by the density operator $T A T^{\dagger}$. Similarly, when performing the von Neumann measurement specified by some $P_1,\ldots,P_l \in \proj(\mathbb{C}^Q)$, the probability that the result of this measurement is $r$ is given by $\Tr(P_r A P_r^{\dagger})$, and if the result is $r$ then the ensemble collapses to an ensemble described by the density operator $\frac{1}{\Tr(P_r A P_r^{\dagger})} P_r A P_r^{\dagger}$. 

Let $V$ and $V'$ denote a pair of finite-dimensional complex Hilbert spaces. Let $\T(V,V')$ denote the $\mathbb{C}$-vector space consisting of all $\mathbb{C}$-linear maps of the form $\Phi:\L(V) \rightarrow \L(V')$. Define $\T(V)=\T(V,V)$ and let $\mathbbm{1}_{\L(V)} \in \T(V)$ denote the identity operator. Consider some $\Phi \in \T(V,V')$. We say that $\Phi$ is \textit{positive} if, $\forall A \in \pos(V)$, we have $\Phi(A) \in \pos(V')$. We say that $\Phi$ is \textit{completely-positive} if, for every finite-dimensional complex Hilbert space $W$, $\Phi \otimes \mathbbm{1}_{\L(W)}$ is positive, where $\otimes$ denotes the tensor product. We say that $\Phi$ is \textit{trace-preserving} if, $\forall A \in \L(V)$, we have $\Tr(\Phi(A))=\Tr(A)$. If $\Phi$ is both completely-positive and trace-preserving, then we say $\Phi$ is a \textit{quantum channel}. Let $\chan(V,V')=\{\Phi \in \T(V,V'):\Phi \text{ is a quantum channel}\}$ denote the set of all such channels, and define $\chan(V)=\chan(V,V)$.

As we wish for our lower bound to be a strong as possible, we wish to consider a variant of the 2QCFA model that is as strong as possible; in particular, we will allow a 2QCFA to perform any physically realizable quantum operation on its quantum register. Following Watrous \cite{watrous2003complexity}, a \textit{selective quantum operation} $\mathcal{E}$ is specified by a set of operators $\{E_{r,j}:r \in R, j \in \{1,\ldots,l\}\} \subseteq \L(\mathbb{C}^Q)$, where $R$ is a finite set and $l \in \mathbb{N}_{\geq 1}$ (throughout the paper, we write $\mathbb{N}_{\geq 1}$ to denote the positive natural numbers, $\mathbb{R}_{\geq 0}$ to denote the nonnegative real numbers, etc.), such that $\sum_{r,j} E_{r,j}^{\dagger}E_{r,j}=\mathbbm{1}_{\mathbb{C}^Q}$. For $r \in R$, we define $\Phi_r \in \T(\mathbb{C}^Q)$ such that, $\Phi_r(A)=\sum_j E_{r,j} A E_{r,j}^{\dagger}$, $\forall A \in \L(V)$. Then, if the quantum register is described by some density operator $A \in \den(\mathbb{C}^Q)$, applying $\mathcal{E}$ will have result $r \in R$ with probability $\Tr(\Phi_r(A))$; if the result is $r$, then the quantum register is described by density operator $\frac{1}{\Tr(\Phi_r(A))} \Phi_r(A)$. Both unitary transformations and von Neumann measurements are special cases of selective quantum operations. For any $\mathcal{E}$, one may always obtain a family of operators that represent $\mathcal{E}$ with $l \leq \lvert Q \rvert^2$ \cite[Theorem 2.22]{watrous2018theory}, and therefore with $l= \lvert Q \rvert^2$ (by defining any extraneous operators to be $\mathbbm{0}_{\mathbb{C}^Q}$). Let $\quantop(\mathbb{C}^Q,R)$ denote the set of all selective quantum operations specified by some $\{E_{r,j}:r \in R, j \in \{1,\ldots,\lvert Q \rvert^2 \}\} \subseteq \L(\mathbb{C}^Q)$.

\subsection{Definition of the 2QCFA Model}\label{sec:prelim:2qcfaDef}

Next, we define two-way finite automata with quantum and classical states (2QCFA), essentially following the original definition of Ambainis and Watrous \cite{ambainis2002two}, with a few alterations that (potentially) make the model stronger. We wish to define the 2QCFA model to be as strong as possible so that our lower bounds against this model are as general as possible. 

Informally, a 2QCFA is a two-way DFA that has been augmented with a quantum register of constant size; the machine may apply unitary transformations to the quantum register and perform (perhaps many) measurements of its quantum register during its computation. Formally, a 2QCFA is a $10$-tuple, $N=(Q,C,\Sigma,R,\theta,\delta,q_{\text{start}},c_{\text{start}},c_{\text{acc}},c_{\text{rej}})$, where $Q$ is a finite set of quantum basis states, $C$ is a finite set of classical states, $\Sigma$ is a finite input alphabet, $R$ is a finite set that specifies the possible results of selective quantum operations, $\theta$ and $\delta$ are the quantum and classical parts of the transition function, $q_{\text{start}}\in Q$ is the quantum start state, $c_{\text{start}} \in C$ is the classical start state, and $c_{\text{acc}},c_{\text{rej}}\in C$, with $c_{\text{acc}} \neq c_{\text{rej}}$, specify the classical accept and reject states, respectively. We define $\#_L,\#_R \not \in \Sigma$, with $\#_L \neq \#_R$, to be special symbols that serve as a left and right end-marker, respectively; we then define the tape alphabet $\Sigma_+=\Sigma \sqcup \{\#_L,\#_R\}$. Let $\widehat{C}=C \setminus \{c_{\text{acc}},c_{\text{rej}}\}$ denote the non-halting classical states. The components of the transition function are as follows: $\theta:\widehat{C} \times \Sigma_+ \rightarrow \quantop(\mathbb{C}^Q,R)$ specifies the selective quantum operation that is to be performed on the quantum register and $\delta:\widehat{C} \times \Sigma_+ \times R \rightarrow C \times \{-1,0,1\}$ specifies how the classical state and (classical) head position evolve.

On an input $w=w_1 \cdots w_n \in \Sigma^*$, with each $w_i \in \Sigma$, the 2QCFA $N$ operates as follows. The machine has a read-only tape that contains the string $\#_L w_1 \cdots w_n \#_R$. Initially, the classic state of $N$ is $c_{\text{start}}$, the quantum register is in the superposition $\ket{q_{\text{start}}}$, and the head is at the left end of the tape, over the left end-marker $\#_L$. On each step of the computation, if the classic state is currently $c \in \widehat{C}$ and the head is over the symbol $\sigma \in \Sigma_+$, $N$ behaves as follows. First, the selective quantum operation $\theta(c,\sigma)$ is performed on the quantum register producing some result $r \in R$. If the result was $r$, and $\delta(c,\sigma,r)=(c',d)$, where $c' \in C$ and $d \in \{-1,0,1\}$, then the classical state becomes $c'$ and the head moves left (resp.  stays put, moves right) if $d=-1$ (resp. $d=0$, $d=1$).

Due to the fact that applying a selective quantum operation is a probabilistic process, the computation of $N$ on an input $w$ is probabilistic. We say that a 2QCFA $N$ recognizes a language $L$ with \textit{two-sided bounded error} $\epsilon$ if, $\forall w \in L$, $\Pr[N \text{ accepts } w] \geq 1-\epsilon$, and, $\forall w \not \in L$, $\Pr[N \text{ accepts } w] \leq \epsilon$. We then define $\mathsf{B2QCFA}(k,d,T(n),\epsilon)$ as the class of languages $L$ for which there is a 2QCFA, with at most $k$ quantum basis states and at most $d$ classical states, that recognizes $L$ with two-sided bounded error $\epsilon$, and has expected running time at most $T(n)$ on all inputs of length at most $n$. In order to make our lower bound as strong as possible, we do \textit{not} require $N$ to halt with probability $1$ on all $w \in \Sigma^*$ (i.e., we permit $N$ to reject an input by looping).

\section{2QCFA Crossing Sequences}\label{sec:crossingSequences}

In this section, we develop a generalization of Hennie's \cite{hennie1965one} notion of crossing sequences to 2QCFA, in which we make use of several ideas from the 2PFA results of Dwork and Stockmeyer \cite{dwork1990time} and Greenberg and Weiss \cite{greenberg1986lower}. This notion will play a key role in our proof of a lower bound on the expected running time of a 2QCFA. 

When a 2QCFA $N=(Q,C,\Sigma,R,\theta,\delta,q_{\text{start}},c_{\text{start}},c_{\text{acc}},c_{\text{rej}})$ is run on an input $w=w_1 \cdots w_n \in \Sigma^*$, where each $w_i \in \Sigma$, the tape consists of $\#_Lw_1 \cdots w_n \#_R$. One may describe the configuration of \textit{a single probabilistic branch} of $N$ at any particular point in time by a triple $(A,c,h)$, where $A \in \den(\mathbb{C}^Q)$ describes the current state of the quantum register, $c \in C$ is the current classical state, and $h \in \{0,\ldots,n+1\}$ is the current head position. To clarify, each step of the computation of $N$ involves applying a selective quantum operation, which is a probabilistic process that produces a particular result $r \in R$ with a certain probability (depending on the operation that is performed and the state of the quantum register); that is to say, the 2QCFA probabilistically branches, with a child for each $r \in R$.

We partition the input as $w=xy$, in some manner to be specified later. We then imagine running $N$ beginning in the configuration $(A,c,\lvert x \rvert)$, where $\lvert x \rvert$ denotes the length of the string $x$ (i.e., the head is initially over the rightmost symbol of $\#_L x$). We wish to describe the configuration (or, more accurately, ensemble of configurations) that $N$ will be in when it ``finishes computing'' on the prefix $\#_L x$, either by ``leaving'' the string $\#_L x$ (by moving its head right when over the rightmost symbol of $\#_L x$), or by accepting or rejecting its input. Of course, $N$ may leave $\#_L x$, then later reenter $\#_L x$, then later leave $\#_L x$ again, and so on, which will naturally lead to our notion of a crossing sequence. Note that the string $y$ does not affect this subcomputation as it occurs entirely within the prefix $\#_L x$.

More generally, we consider the case in which $N$ is run on the prefix $\#_L x$, where $N$ starts in some ensemble of configurations $\{(p_i,(A_i,c_i,\lvert x \rvert)):i \in I\}$, where the probability of being in configuration $(A_i,c_i,\lvert x \rvert)$ is given by $p_i$ (note that the head position in each configuration is over the rightmost symbol of $\#_L x$); we call this ensemble a \textit{starting ensemble}. We then wish to describe the ensemble of configurations that $N$ will be in when it ``finishes computing'' on the prefix $\#_L x$, (essentially) as defined above; we call this ensemble a \textit{stopping ensemble}\footnote{We use the terms ``starting ensemble'' and ``stopping ensemble'' to make clear the similarity to the notion of a  ``starting condition'' and of a ``stopping condition'' used by  Dwork and Stockmeyer \cite{dwork1990time} in their 2PFA result.}. Much as it was the case that an ensemble of pure states of a quantum register can be described by a density operator, we may also describe an ensemble of configurations of a 2QCFA using density operators. This will greatly simplify our definition and analysis of the crossing sequence of a 2QCFA.

\subsection{Describing Ensembles of Configurations of 2QCFA}\label{sec:crossingSequences:configs}

The 2QCFA $N$ posseses both a constant-sized \textit{quantum register}, that is described by some density operator at any particular point in time, and a constant-sized \textit{classical register}, that stores a classical state $c \in C$. We can naturally interpret each $c \in C$ as an element $\ket{c} \in \mathbb{C}^C$, of a special type; that is to say, each classical state $c$ corresponds to some element $\ket{c}$ in the natural orthonormal basis of $\mathbb{C}^C$ (whereas each superposition $\ket{\psi}$ of the quantum register corresponds to an element of $\mathbb{C}^Q$ of norm $1$). One may also view $N$ as possessing a \textit{head register} that stores a (classical) head position $h \in H_x=\{0,\ldots,\lvert x \rvert+1\}$ (when computing on the prefix $\#_L x$); of course, the size of this pseudo-register grows with the input prefix $x$. We analogously interpret a head position $h \in H_x$ as being the ``classical'' element $\ket{h} \in \mathbb{C}^{H_x}$. A configuration $(A,c,h) \in \den(\mathbb{C}^Q) \times C \times H_x$ is then simply a state of the \textit{combined register}, which consists of the quantum, classical, and head registers. 

We then consider an \textit{ensemble of configurations} $\{(p_i,(A_i,c_i,h_i)):i \in I\}$, where $p_i$ denotes the probability of being in configuration $(A_i,c_i,h_i)$. We represent this ensemble (non-uniquely) by the density operator $Z=\sum_i \big(p_iA_i \otimes \ket{c_i}\bra{c_i} \otimes \ket{h_i}\bra{h_i}\big) \in \den(\mathbb{C}^Q \otimes \mathbb{C}^C \otimes \mathbb{C}^{H_x})$. Let $\widehat{i}(c,h)=\{i\in I:(c_i,h_i)=(c,h)\}$ denote the indices of those configurations in classical state $c$ and with head position $h$. We then define $p:C \times H_x \rightarrow [0,1]$ such that $p(c,h)=\sum_{i \in \widehat{i}(c,h)} p_i$ is the total probability of being in classical state $c$ and having head position $h$. We define $A:C \times H_x \rightarrow \den(\mathbb{C}^Q)$ such that, if $p(c,h) \neq 0$, then $A(c,h)=\sum_{i \in \widehat{i}(c,h)} \frac{p_i}{p(c,h)} A_i$ is the density operator obtained by ``merging'' all density operators $A_i$ that come from configurations $(A_i,c_i,h_i)$ with classical state $c_i=c$ and head position $h_i=h$; if $p(c,h)=0$, then we define $A(c,h)$ arbitrarily. Then $Z=\sum_{c,h} \big(p(c,h) A(c,h) \otimes \ket{c}\bra{c} \otimes \ket{h}\bra{h}\big)$. Let $\widehat{\den}(\mathbb{C}^Q \otimes \mathbb{C}^C \otimes \mathbb{C}^{H_x})$ denote the set of all density operators given by some $Z$ of the above form  (i.e., those density operators that respect the fact that both the classical state and head position are classical).

We also consider the case in which we are only interested in the states of the quantum and classical registers, but not the head position. We then analogously describe an ensemble $\{(p_i,(A_i,c_i)):i \in I\}$ by $Z=\sum_i \big(p_iA_i \otimes \ket{c_i}\bra{c_i}\big)\in \den(\mathbb{C}^Q \otimes \mathbb{C}^C)$, and we define $\widehat{\den}(\mathbb{C}^Q \otimes \mathbb{C}^C)$ to be the set of all such density operators. In a starting ensemble, all configurations have the same head position: $\lvert x \rvert$. We define $I_x \in \T(\mathbb{C}^Q \otimes \mathbb{C}^C, \mathbb{C}^Q \otimes \mathbb{C}^C \otimes \mathbb{C}^{H_x})$ such that $I_x(Z)=Z \otimes \ket{\lvert x \rvert}\bra{\lvert x \rvert}$. Similarly, in a stopping ensemble, all configurations either have head position $\lvert x \rvert+1$ or are accepting or rejecting configurations (in which the head position is irrelevant). Let $\Tr_{\mathbb{C}^{H_x}}=\mathbbm{1}_{\L(\mathbb{C}^Q \otimes \mathbb{C}^C)} \otimes \Tr \in \T(\mathbb{C}^Q \otimes \mathbb{C}^C \otimes \mathbb{C}^{H_x},\mathbb{C}^Q \otimes \mathbb{C}^C)$ denote the \textit{partial trace with respect to} $\mathbb{C}^{H_x}$.

\subsection{Definition and Properties of 2QCFA Crossing Sequences}\label{sec:crossingSequences:def}

We now formally define the notion of a crossing sequence of a 2QCFA and prove certain needed properties. We begin by establishing some notation.

\begin{definition2}\label{def:2qcfa:transitionDesc}
	Consider a 2QCFA $N=(Q,C,\Sigma,R,\theta,\delta,q_{\text{start}},c_{\text{start}},c_{\text{acc}},c_{\text{rej}})$. For $c \in \widehat{C}=C \setminus \{c_{\text{acc}},c_{\text{rej}}\}$, $\sigma \in \Sigma_+=\Sigma\sqcup \{\#_L,\#_R\}$, $r \in R$, and $j \in J=\{1,\ldots,\lvert Q \rvert^2\}$, we make the following definitions.
	\begin{enumerate}[(i)]
		\item\label{def:2qcfa:transitionDesc:quantOp} Define $E_{c,\sigma,r,j} \in \L(\mathbb{C}^Q)$ such that $\theta(c,\sigma)\in \quantop(\mathbb{C}^Q,R)$ is described by $\{E_{c,\sigma,r,j}:r \in R, j \in J\}$.
		\item\label{def:2qcfa:transitionDesc:partialQuantOp} Define $\Phi_{c,\sigma,r} \in \T(\mathbb{C}^Q)$ such that $\Phi_{c,\sigma,r}(A)=\sum_j E_{c,\sigma,r,j} A E_{c,\sigma,r,j}^{\dagger}$, $\forall A \in \L(\mathbb{C}^Q)$.
		\item\label{def:2qcfa:transitionDesc:classicOp} Let $\gamma_{c,\sigma,r} \in C$ and $d_{c,\sigma,r} \in \{-1,0,1\}$ denote, respectively, the new classical state and the motion of the head, if the result of applying $\theta(c,\sigma)$ is $r$; i.e., $\delta(c,\sigma,r)=(\gamma_{c,\sigma,r},d_{c,\sigma,r})$. 
	\end{enumerate}  
\end{definition2}

Consider some $x \in \Sigma^*$. Let $\widehat{H}_x=\{0,\ldots,\lvert x \rvert \}$ denote the head positions corresponding to the prefix $\#_L x$, and let $H_x=\{0,\ldots,\lvert x \rvert+1\}$ denote the set of possible positions the head of $N$ may be in until it ``finishes computing'' on the prefix $\#_L x$. We define an operator $S_x \in \T(\mathbb{C}^Q \otimes \mathbb{C}^{C} \otimes \mathbb{C}^{H_x})$ that describes a single step of the computation of $N$ on $\#_L x$, as follows. If $(c,h) \in \widehat{C} \times \widehat{H}_x$, then $S_x(A \otimes \ket{c}\bra{c} \otimes \ket{h}\bra{h})$ describes the ensemble of configurations of $N$ after running $N$ for a single step beginning in the configuration $(A,c,h)$; otherwise (i.e., if $c \in \{c_{\text{acc}},c_{\text{rej}}\}$ or $h=\lvert x \rvert+1$, which means $N$ has ``finished computing'' on $\#_L x$) $S_x$ leaves the configuration unchanged. We will observe that $S_x$ correctly describes the behavior of $N$ on an ensemble of configurations, and that $S_x$ is a quantum channel.

\begin{definition2}\label{def:2qcfa:singleStepOp}
	Using the notation of Definition~\ref{def:2qcfa:transitionDesc}, consider a 2QCFA $N$ and a string $x \in \Sigma^*$. Let $x_h \in \Sigma$ denote the symbol of $x$ at position $h$, and let $x_0=\#_L$ denote the left end-marker. 
	\begin{enumerate}[(i)]
		\item\label{def:2qcfa:singleStepOp:newKraus} For $(c,h,r,j) \in C \times H_x \times R \times J$, define $\widetilde{E}_{x,c,h,r,j} \in \L(\mathbb{C}^Q \otimes \mathbb{C}^{C} \otimes \mathbb{C}^{H_x})$ as follows. 
		\[\widetilde{E}_{x,c,h,r,j}=\begin{cases}
		E_{c,x_h,r,j} \otimes \ket{\gamma_{c,x_h,r}}\bra{c} \otimes \ket{h+d_{c,x_h,r}}\bra{h}, & \text{if } (c,h) \in \widehat{C} \times \widehat{H}\\
		\frac{1}{\sqrt{\lvert R \rvert \lvert J \rvert}}\mathbbm{1}_{\mathbb{C}^Q} \otimes \ket{c}\bra{c} \otimes \ket{h}\bra{h}, & \text{otherwise.}\end{cases}\]		 
		\item\label{def:2qcfa:singleStepOp:OpDef} Define $S_x \in \T(\mathbb{C}^Q \otimes \mathbb{C}^{C} \otimes \mathbb{C}^{H_x})$ such that \[S_x(Z)=\sum_{(c,h,r,j) \in C \times H_x \times R \times J} \widetilde{E}_{x,c,h,r,j} Z \widetilde{E}_{x,c,h,r,j}^{\dagger}, \ \ \forall Z \in \L(\mathbb{C}^Q \otimes \mathbb{C}^{C} \otimes \mathbb{C}^{H_x}).\]
	\end{enumerate}  
\end{definition2}

\begin{lemma}\label{thm:2qcfa:singleStepOp:correctlyDescribesNonHalt}
	Using the above notation, consider some $x \in \Sigma^*$ and $(A,\widehat{c},\widehat{h}) \in \den(\mathbb{C}^Q) \times \widehat{C} \times \widehat{H}_x$. Let $\widehat{Z}=A \otimes \ket{\widehat{c}}\bra{\widehat{c}} \otimes \ket{\widehat{h}}\bra{\widehat{h}}$. $S_x(\widehat{Z})$ describes the ensemble of configurations obtained after running $N$ for one step, beginning in the configuration $(A,\widehat{c},\widehat{h})$, on input prefix $\#_L x$.
\end{lemma}
\begin{proof}
	Let $\widetilde{R}_{x,\widehat{c},\widehat{h},A}=\{r \in R: \Tr(\Phi_{\widehat{c},x_{\widehat{h}},r}(A))\neq 0\}$. Note that $A \in \den(\mathbb{C}^Q) \subseteq \pos(\mathbb{C}^Q)$, which implies $\Phi_{\widehat{c},x_{\hat{h}},r}(A) \in \pos(\mathbb{C}^Q)$; therefore, we have $\Tr(\Phi_{\widehat{c},x_{\hat{h}},r}(A))=0$ precisely when $\Phi_{\widehat{c},x_{\hat{h}},r}(A)=\mathbbm{0}_{\mathbb{C}^Q}$. After running $N$ as described, it is in an ensemble of configurations \[\left\{\left(\Tr(\Phi_{\widehat{c},x_{\hat{h}},r}(A)),\left(\frac{1}{\Tr(\Phi_{\widehat{c},x_{\hat{h}},r}(A))} \Phi_{\widehat{c},x_{\hat{h}},r}(A),\gamma_{\widehat{c},x_{\hat{h}},r},\widehat{h}+d_{\widehat{c},x_{\hat{h}},r}\right)\right): r \in \widetilde{R}_{x,\widehat{c},\widehat{h},A} \right\}.\]
	This ensemble of configurations is described by the density operator $\widehat{Z}'$ given by
	\[\widehat{Z}'=\sum_{r \in \widetilde{R}_{x,\widehat{c},\widehat{h},A}} \left( \frac{\Tr(\Phi_{\widehat{c},x_{\hat{h}},r}(A))}{\Tr(\Phi_{\widehat{c},x_{\hat{h}},r}(A))}\Phi_{\widehat{c},x_{\hat{h}},r}(A) \otimes \ket{\gamma_{\widehat{c},x_{\hat{h}},r}}\bra{\gamma_{\widehat{c},x_{\hat{h}},r}} \otimes \ket{\widehat{h}+d_{\widehat{c},x_{\hat{h}},r}}\bra{\widehat{h}+d_{\widehat{c},x_{\hat{h}},r}} \right) \]
	\[=\sum_{r \in R} \left( \Phi_{\widehat{c},x_{\hat{h}},r}(A) \otimes \ket{\gamma_{\widehat{c},x_{\hat{h}},r}}\bra{\gamma_{\widehat{c},x_{\hat{h}},r}} \otimes \ket{\widehat{h}+d_{\widehat{c},x_{\hat{h}},r}}\bra{\widehat{h}+d_{\widehat{c},x_{\hat{h}},r}} \right).\]	
	Let $B_{x,\widehat{c},\widehat{h},r}=\ket{\gamma_{\widehat{c},x_{\hat{h}},r}}\bra{\gamma_{\widehat{c},x_{\hat{h}},r}} \otimes \ket{\widehat{h}+d_{\widehat{c},x_{\hat{h}},r}}\bra{\widehat{h}+d_{\widehat{c},x_{\hat{h}},r}}$. If $(c,h) \in \widehat{C} \times \widehat{H}_x$, then 
	\[\widetilde{E}_{x,c,h,r,j}\widehat{Z}\widetilde{E}_{x,c,h,r,j}^{\dagger}= \widetilde{E}_{x,c,h,r,j} \left(A \otimes \ket{\widehat{c}}\bra{\widehat{c}} \otimes \ket{\widehat{h}}\bra{\widehat{h}}\right) \widetilde{E}_{x,c,h,r,j}^{\dagger}\]
	\[=E_{c,x_h,r,j}A E_{c,x_h,r,j}^{\dagger} \otimes \ket{\gamma_{c,x_h,r}}\braket{c}{\widehat{c}}\braket{\widehat{c}}{c}\bra{\gamma_{c,x_h,r}} \otimes \ket{h+d_{c,x_h,r}}\braket{h}{\widehat{h}}\braket{\widehat{h}}{h}\bra{h+d_{c,x_h,r}}\]
	\[=\begin{cases}
	E_{\widehat{c},x_{\hat{h}},r,j}A E_{\widehat{c},x_{\hat{h}},r,j}^{\dagger} \otimes B_{x,\widehat{c},\widehat{h},r}, & \text{if } (c,h)=(\widehat{c},\widehat{h}) \\
	\mathbbm{0}_{\mathbb{C}^Q \otimes \mathbb{C}^{C} \otimes \mathbb{C}^{H_x}}, & \text{otherwise.}
	\end{cases}\]
	If, instead, $(c,h) \not \in \widehat{C} \times \widehat{H}_x$, then $\widetilde{E}_{x,c,h,r,j}\widehat{Z}\widetilde{E}_{x,c,h,r,j}^{\dagger}=\mathbbm{0}_{\mathbb{C}^Q \otimes \mathbb{C}^{C} \otimes \mathbb{C}^{H_x}}$. Therefore
	\[S_x(\widehat{Z})= \sum_{(r,j) \in R \times J} \sum_{(c,h) \in C \times H_x} \widetilde{E}_{x,c,h,r,j} \widehat{Z} \widetilde{E}_{x,c,h,r,j}^{\dagger}= \sum_{(r,j) \in R \times J} \left(E_{\widehat{c},x_{\hat{h}},r,j}A E_{\widehat{c},x_{\hat{h}},r,j}^{\dagger} \otimes B_{x,\widehat{c},\widehat{h},r}\right)\]
	\[=\sum_{r \in R}\bigg( \bigg(\sum_{j \in J} E_{\widehat{c},x_{\hat{h}},r,j}A E_{\widehat{c},x_{\hat{h}},r,j}^{\dagger}\bigg) \otimes B_{x,\widehat{c},\widehat{h},r}\bigg)=\sum_{r \in R}\left( \Phi_{\widehat{c},x_{\hat{h}},r}(A) \otimes B_{x,\widehat{c},\widehat{h},r}\right)=\widehat{Z}'.\qedhere \] 
\end{proof}

\begin{lemma}\label{thm:2qcfa:singleStepOp:correctlyDescribesEnsemble}
	Consider some $x \in \Sigma^*$ and $Z \in \widehat{\den}(\mathbb{C}^Q \otimes \mathbb{C}^C \otimes \mathbb{C}^{H_x})$. If $\{(p_i,(A_i,c_i,h_i)):i \in I\}$ is some ensemble of configurations described by $Z$, then $S_x(Z)$ describes the ensemble of configurations obtained by replacing each configuration with $(c_i,h_i) \in (\widehat{C} \times \widehat{H}_x)$ by the ensemble (scaled by $p_i$) of configurations obtained by running $N$ for one step beginning in the configuration $(A_i,c_i,h_i)$, and leaving each configuration with $(c_i,h_i) \not \in (\widehat{C} \times \widehat{H}_x)$ unchanged. 
\end{lemma}
\begin{proof} 
	This follows immediately from Lemma~\ref{thm:2qcfa:singleStepOp:correctlyDescribesNonHalt} and linearity.
\end{proof}

\begin{lemma}\label{thm:2qcfa:singleStepOp:isChannel}
	$S_x \in \chan(\mathbb{C}^Q \otimes \mathbb{C}^C \otimes \mathbb{C}^{H_x})$, $\forall x \in \Sigma^*$.
\end{lemma}
\begin{proof}
	$\{\widetilde{E}_{x,c,h,r,j}:(c,h,r,j) \in C \times H_x \times R \times J\}$ is a \textit{Kraus representation} of $S_x$; therefore, $S_x \in \chan(\mathbb{C}^Q \otimes \mathbb{C}^C \otimes \mathbb{C}^{H_x}) \Leftrightarrow \sum\limits_{c,h,r,j} \widetilde{E}_{x,c,h,r,j}^{\dagger} \widetilde{E}_{x,c,h,r,j}=\mathbbm{1}$ \cite[Corollary 2.27]{watrous2018theory}. 
	
	We begin by showing that $\sum_{r,j}\widetilde{E}_{x,c,h,r,j}^{\dagger} \widetilde{E}_{x,c,h,r,j}=\mathbbm{1}_{\mathbb{C}^Q}\otimes \ket{c}\bra{c} \otimes \ket{h}\bra{h}$, $\forall (c,h) \in C \times H_x$. First, suppose $(c,h) \in \widehat{C} \times \widehat{H}_x$; we then have 
	$$\widetilde{E}_{x,c,h,r,j}^{\dagger} \widetilde{E}_{x,c,h,r,j}=\big(E_{c,x_h,r,j}^{\dagger} \otimes \ket{c}\bra{\gamma_{c,x_h,r}} \otimes \ket{h}\bra{h+d_{c,x_h,r}}\big)\big(E_{c,x_h,r,j} \otimes \ket{\gamma_{c,x_h,r}}\bra{c} \otimes \ket{h+d_{c,x_h,r}}\bra{h}\big)$$
	$$=E_{c,x_h,r,j}^{\dagger}E_{c,x_h,r,j} \otimes \ket{c}\braket{\gamma_{c,x_h,r}}{\gamma_{c,x_h,r}}\bra{c} \otimes \ket{h}\braket{h+d_{c,x_h,r}}{h+d_{c,x_h,r}}\bra{h}$$
	$$=E_{c,x_h,r,j}^{\dagger}E_{c,x_h,r,j} \otimes \ket{c}\bra{c} \otimes \ket{h}\bra{h}.$$
	This implies, 
	$$\sum_{(r,j) \in R \times J} \widetilde{E}_{x,c,h,r,j}^{\dagger} \widetilde{E}_{x,c,h,r,j}= \bigg( \sum_{(r,j) \in R \times J}  E_{c,x_h,r,j}^{\dagger}E_{c,x_h,r,j} \bigg)  \otimes \ket{c}\bra{c} \otimes \ket{h}\bra{h} = \mathbbm{1}_{\mathbb{C}^Q}\otimes \ket{c}\bra{c} \otimes \ket{h}\bra{h}.$$
	If, instead, $(c,h) \not \in \widehat{C} \times \widehat{H}_x$, then $$ \widetilde{E}_{x,c,h,r,j}^{\dagger} \widetilde{E}_{x,c,h,r,j}=  \bigg(\frac{1}{\sqrt{\lvert R \rvert \lvert J \rvert}}\mathbbm{1}_{\mathbb{C}^Q} \otimes \ket{c}\bra{c} \otimes \ket{h}\bra{h}\bigg)^{\dagger} \bigg(\frac{1}{\sqrt{\lvert R \rvert \lvert J \rvert}}\mathbbm{1}_{\mathbb{C}^Q} \otimes \ket{c}\bra{c} \otimes \ket{h}\bra{h}\bigg)$$ $$= \frac{1}{\lvert R \rvert \lvert J \rvert}\mathbbm{1}_{\mathbb{C}^Q} \otimes \ket{c}\bra{c} \otimes \ket{h}\bra{h}.$$
	This implies $$\sum_{(r,j) \in R \times J} \widetilde{E}_{x,c,h,r,j}^{\dagger} \widetilde{E}_{x,c,h,r,j}= \mathbbm{1}_{\mathbb{C}^Q}\otimes \ket{c}\bra{c} \otimes \ket{h}\bra{h}.$$	
	We then have \[\sum_{(c,h) \in C \times H_x}\sum_{(r,j) \in R \times J} \widetilde{E}_{x,c,h,r,j}^{\dagger} \widetilde{E}_{x,c,h,r,j}=\sum_{(c,h) \in C \times H_x} \mathbbm{1}_{\mathbb{C}^Q}\otimes \ket{c}\bra{c} \otimes \ket{h}\bra{h}= \mathbbm{1}_{\mathbb{C}^Q \otimes \mathbb{C}^C \otimes \mathbb{C}^{H_x}}.\qedhere\]

\end{proof}

For $m \in \mathbb{N}$, we define the $m$-\textit{truncated stopping ensemble} as the ensemble of configurations that $N$ will be in when it ``finishes computing'' on $\#_L x$, as defined earlier, with the modification that if any particular branch of $N$ runs for more than $m$ steps, the computation of that branch will be ``interrupted'' immediately before it attempts to perform the $m+1^{\text{st}}$ step and instead immediately reject. To be clear, this truncation occurs only in the \textit{analysis} of $N$; we do not modify the 2QCFA. The following truncation operator $T_x$, which terminates all branches on which $N$ has not yet ``finished computing,'' will help us do this.

\begin{definition2}\label{def:2qcfa:truncOp}
	 For $(c,h) \in (C,H_x)$, let $\widehat{E}_{x,c,h}=\mathbbm{1}_{\mathbb{C}^Q} \otimes \ket{c'}\bra{c} \otimes \ket{h}\bra{h}$, where $c'=c_{\text{rej}}$ if $(c,h) \in \widehat{C} \times \widehat{H}_x$, and $c'=c$ otherwise. We then define $T_x \in \T(\mathbb{C}^Q \otimes \mathbb{C}^C \otimes \mathbb{C}^{H_x})$ such that $T_x(Z)=\sum_{(c,h) \in C \times H_x}  \widehat{E}_{x,c,h} Z \widehat{E}_{x,c,h}^{\dagger}$. 
\end{definition2}

\begin{lemma}\label{thm:2qcfa:truncOpProp}
	Using the above notation, the following statements hold.
	\begin{enumerate}[(i)]
		\item\label{thm:2qcfa:truncOpProp:correctlyDescribes} For any $Z \in \widehat{\den}(\mathbb{C}^Q \otimes \mathbb{C}^C \otimes \mathbb{C}^{H_x})$, if $\{(p_i,(A_i,c_i,h_i)):i \in I\}$ is any ensemble of configurations described by $Z$, then $T_x(Z)$ describes the ensemble of configurations in which each configuration with $(c_i,h_i) \in \widehat{C} \times \widehat{H}_x$ is replaced by the configuration $(A_i,c_{\text{rej}},h_i)$ (i.e., all configurations in which $N$ has not yet ``finished computing'' on $\#_L x$ become rejecting configurations) and all other configurations are left unchanged. 
		\item\label{thm:2qcfa:truncOpProp:isChannel} $T_x \in \chan(\mathbb{C}^Q \otimes \mathbb{C}^C \otimes \mathbb{C}^{H_x})$.
	\end{enumerate}
\end{lemma}
\begin{proof}
	\begin{enumerate}[(i)]
		\item Immediate from definitions.	
		\item As in the proof of Lemma~\ref{thm:2qcfa:singleStepOp:isChannel}, we may straightforwardly show $\sum_{c,h} \widehat{E}_{x,c,h}^{\dagger} \widehat{E}_{x,c,h}=\mathbbm{1}_{\mathbb{C}^Q \otimes \mathbb{C}^C \otimes \mathbb{C}^{H_x}}$, which implies $T_x \in \chan(\mathbb{C}^Q \otimes \mathbb{C}^C \otimes \mathbb{C}^{H_x})$ \cite[Corollary 2.27]{watrous2018theory}. \qedhere
	\end{enumerate}
\end{proof}

The following operator converts starting ensembles to $m$-truncated stopping ensembles.

\begin{definition2}\label{def:2qcfa:truncTranOp} 
	For $x \in \Sigma^*$ and $m \in \mathbb{N}$, we define the $m$-\textit{truncated transfer operator} $N_{x,m}^{\Balancedhookarrowright}=\Tr_{\mathbb{C}^{H_x}} \circ T_x \circ S_x^m \circ I_x \in \T(\mathbb{C}^Q \otimes \mathbb{C}^C)$. For $y \in \Sigma^*$, we next consider the ``dual case'' of running $N$ on the suffix $y \#_R$ beginning in some ensemble of configurations $\{(p_i,(A_i,c_i,\lvert x \rvert+1)):i \in I\}$ (i.e., the head position of every configuration is over the leftmost symbol of $y \#_R$). We define the notion of an $m$-truncated stopping ensemble, and all other notions, symmetrically. That is to say, a branch of $N$ ``finishes computing'' on $y \#_R$ when it either ``leaves'' $y \#_R$ (by moving its head left from the leftmost symbol of $y \#_R$), or accepts or rejects the input, or runs for more than $m$ steps. We then define $N_{y,m}^{\Balancedhookarrowleft}\in \T(\mathbb{C}^Q \otimes \mathbb{C}^C)$ as the corresponding ``dual'' $m$-truncated transfer operator for $y$.
\end{definition2}

\begin{lemma}\label{thm:2qcfa:truncTranOpAndCrossProp} 
	Using the notation of Definition~\ref{def:2qcfa:truncTranOp}, the following statements hold.
	\begin{enumerate}[(i)]
		\item\label{thm:2qcfa:truncTranOpAndCrossProp:primalCorrectlyDescribes} For $Z \in \widehat{\den}(\mathbb{C}^Q \otimes \mathbb{C}^C)$, if $N$ is run on $\#_L x$ beginning in any ensemble of configurations described by $I_x(Z)$ (i.e., the head position of every configuration is over the rightmost symbol of $\#_L x$), then the $m$-truncated stopping ensemble is described by $N_{x,m}^{\Balancedhookarrowright}(Z)$. 
		\item\label{thm:2qcfa:truncTranOpAndCrossProp:dualCorrectlyDescribes} For $Z \in \widehat{\den}(\mathbb{C}^Q \otimes \mathbb{C}^C)$, if $N$ is run on $y \#_R$ beginning in any ensemble of configurations described by $I_{x+1}(Z)$, then the $m$-truncated stopping ensemble is described by $N_{y,m}^{\Balancedhookarrowleft}(Z)$.
		\item\label{thm:2qcfa:truncTranOpAndCrossProp:isChannel} We have $N_{x,m}^{\Balancedhookarrowright}, N_{y,m}^{\Balancedhookarrowleft} \in \chan(\mathbb{C}^Q \otimes \mathbb{C}^C)$, $\forall x,y \in \Sigma^*$, $\forall m \in \mathbb{N}$.
	\end{enumerate}
\end{lemma}
\begin{proof}
	\begin{enumerate}[(i)]
		\item Immediate by \Cref{def:2qcfa:truncTranOp}, \Cref{thm:2qcfa:singleStepOp:correctlyDescribesEnsemble}, and \Cref{thm:2qcfa:truncOpProp}(\ref{thm:2qcfa:truncOpProp:correctlyDescribes}).
		
		\item Immediate by \Cref{def:2qcfa:truncTranOp}, and analogous versions of \Cref{thm:2qcfa:singleStepOp:correctlyDescribesEnsemble}, and \Cref{thm:2qcfa:truncOpProp}(\ref{thm:2qcfa:truncOpProp:correctlyDescribes}).
		
		\item By definition, $N_{x,m}^{\Balancedhookarrowright}=\Tr_{\mathbb{C}^{H_x}} \circ T_x \circ S_x^m \circ I_x$. By \Cref{thm:2qcfa:singleStepOp:isChannel} and \Cref{thm:2qcfa:truncOpProp}(\ref{thm:2qcfa:truncOpProp:isChannel}), we have $S_x,T_x \in \chan(\mathbb{C}^Q \otimes \mathbb{C}^C \otimes \mathbb{C}^{H_x})$. It is straightforward to see that $I_x \in \chan(\mathbb{C}^Q \otimes \mathbb{C}^C, \mathbb{C}^Q \otimes \mathbb{C}^C \otimes \mathbb{C}^{H_x})$ and $\Tr_{\mathbb{C}^{H_x}} \in \chan(\mathbb{C}^Q \otimes \mathbb{C}^C \otimes \mathbb{C}^{H_x},\mathbb{C}^Q \otimes \mathbb{C}^C)$ and that the composition of quantum channels is a quantum channel (see, for instance, \cite[Section 2.2]{watrous2018theory}). The claim for $N_{y,m}^{\Balancedhookarrowleft}$ follows by an analogous argument. \qedhere
	\end{enumerate}
\end{proof}

Given a 2QCFA $N$, we produce an equivalent $N'$ of a certain convenient form, in much the same way that Dwork and Stockmeyer \cite{dwork1990time} converted a 2PFA to a convenient form. The 2QCFA $N'$ is identical to $N$, except for the addition of two new classical states, $c_{\text{start}}'$ and $c'$, where $c_{\text{start}}'$ will be the start state of $N'$. On any input, $N'$ will move its head to the right until it reaches $\#_R$, performing the trivial transformation to its quantum register along the way. When it reaches $\#_R$, $N'$ will enter $c'$; then, $N'$ will move its head to the left until it reaches $\#_L$, again performing the trivial transformation to its quantum register. When it reaches $\#_L$, $N'$ will enter the original start state $c_{\text{start}}$ and behave identically to $N$ from this point. For the remainder of the paper, we assume all 2QCFA have this form. 

Finally, we define the $m$-\textit{truncated crossing sequence}.

\begin{definition2}\label{def:2qcfa:truncCross} 
	For $x,y \in \Sigma^*$ and $m \in \mathbb{N}$, the $m$-\textit{truncated crossing sequence} of $N$ with respect to the (partitioned) input $xy$ is the sequence $Z_1,Z_2,\ldots \in \widehat{\den}(\mathbb{C}^Q \otimes \mathbb{C}^C)$, defined as follows. The density operator $Z_1$ describes the ensemble consisting of the single configuration (of the quantum register and classical register) $(\ket{q_{\text{start}}},c_{\text{start}})$ that $N$ is in when it first crosses from $\#_L x$ into $y \#_R$, which is of this simple form due to the assumed form of $N$. The sequence $Z_1,Z_2,\ldots$ is then obtained by starting with $Z_1$ and alternately applying $N_{y,m}^{\Balancedhookarrowleft}$ and $N_{x,m}^{\Balancedhookarrowright}$. To be precise, \[Z_i=\begin{cases} 
	\ket{q_{\text{start}}}\bra{q_{\text{start}}} \otimes \ket{c_{\text{start}}}\bra{c_{\text{start}}}, & i=1\\
	N_{y,m}^{\Balancedhookarrowleft}(Z_{i-1}), & i>1, i \text{ is even}\\
	N_{x,m}^{\Balancedhookarrowright}(Z_{i-1}), & i>1, i \text{ is odd.}
	\end{cases}\]
\end{definition2}

\begin{remark}
	Note that the $\{Z_i\}$ that comprise a crossing sequence do \textit{not} describe the ensemble of configurations of $N$ at particular points in time during its computation on the input $xy$; instead, $Z_i$ describes the ensemble of configurations of the set of all the probabilistic branches of $N$ at the $i^{\text{th}}$ time each branch crosses between $\#_L x$ and $y \#_R$(with the convention stated above of considering a branch that has accepting or rejected its input to ``cross'' in classic state $c_{\text{acc}}$ or $c_{\text{rej}}$, respectively, indefinitely; as well as the convention that if a given branch of $N$ attempts to run for more than $m$ steps within the prefix $\#_L x$ or within the suffix $y\#_R$, that branch is interrupted and immediately forced to reject). Of course, a given branch may not cross between $\#_L x$ and $y \#_R$ more than $i$ times within the first $i$ steps of the computation; this will allow us to use such crossing sequences to prove a lower bound on the expected running time of $N$. 
\end{remark}

\begin{remark}
	Moreover, while the $m$-truncated crossing operator $N_{x,m}^{\Balancedhookarrowright}$ completely suffices for our analysis, one could also define a \textit{non-truncated transfer operator} $N_x^{\Balancedhookarrowright} \in \chan(\mathbb{C}^Q \otimes \mathbb{C}^C)$ as an accumulation point of the sequence $(N_{x,m}^{\Balancedhookarrowright})_{m \in \mathbb{N}}$; such an accumulation point exists due to the fact that $\chan(\mathbb{C}^Q \otimes \mathbb{C}^C)$ is compact \cite[Proposition 2.28]{watrous2018theory}. Using $N_x^{\Balancedhookarrowright}$ and the symmetrically defined $N_y^{\Balancedhookarrowleft}$, one could then define the \textit{non-truncated crossing sequence} of $N$ on $xy$. The resulting analyses of these two types of crossing sequences would essentially be identical, and so we do not consider this definition further here; however, the (somewhat cleaner) non-truncated crossing sequence may be more useful in other applications. 
\end{remark}

\section{Lower Bounds on the Running Time of 2QCFA}\label{sec:runningTimeLowerBound}

Dwork and Stockmeyer proved a lower bound \cite[Lemma 4.3]{dwork1990time} on the expected running time $T(n)$ of any 2PFA that recognizes any language $L$ with bounded error, in terms of their hardness measure $D_L(n)$. We prove that an analogous claim holds for any 2QCFA. The preceding quantum generalization of a crossing sequence plays a key role in the proof, essentially taking the place of the Markov chains used both in the aforementioned result of Dwork and Stockmeyer and in the earlier result of Greenberg and Weiss \cite{greenberg1986lower}, which showed that 2PFA cannot recognize $L_{eq}$ in subexponential time. 

\subsection{Nonregularity}\label{sec:runningTimeLowerBound:nonregularity}
For a language $L$, Dwork and Stockmeyer \cite{dwork1990time} defined a particular ``hardness measure'' $D_L:\mathbb{N} \rightarrow \mathbb{N}$, which they called the \textit{nonregularity} of $L$, as follows. Let $\Sigma$ be a finite alphabet, $L \subseteq \Sigma^*$ a language, and $n \in \mathbb{N}$. Let $\Sigma^{\leq n}=\{w \in \Sigma^*:\lvert w \rvert \leq n\}$ denote the set of all strings over $\Sigma$ of length at most $n$ and consider some $x,x' \in \Sigma^{\leq n}$. We say that $x$ and $x'$ are $(L,n)$-\textit{dissimilar}, which we denote by writing $x \not \sim_{L,n} x'$,  if $\exists y \in \Sigma^{\leq n-\max(\lvert x \rvert,\lvert x' \rvert}$, such that $xy \in L \Leftrightarrow x'y \not \in L$. Recall the classic Myhill-Nerode inequivalence relation, in which $x,x' \in \Sigma^*$ are $L$-dissimilar if $\exists y \in \Sigma^*$, such that $xy \in L \Leftrightarrow x'y \not \in L$. Then $x,x' \in \Sigma^{\leq n}$ are $(L,n)$-dissimilar precisely when they are $L$-dissimilar, and the dissimilarity is witnessed by a ``short'' string $y$.  We then define $D_L(n)$ to be the largest $h \in \mathbb{N}$ such that $\exists x_1,\ldots,x_h \in \Sigma^{\leq n}$ that are pairwise $(L,n)$-dissimilar (i.e., $x_i \not \sim_{L,n} x_j$, $\forall i,j$ with $i \neq j$).

In fact, $D_L$ has been defined by many authors, both before and after Dwork and Stockmeyer, who gave many different names to this quantity and who (repeatedly) rediscovered certain basic facts about it; we refer the reader to the excellent paper of Shallit and Breitbart \cite{shallit1996automaticityI} for a detailed history of the study of $D_L$ and related hardness measures. In the remainder of this section, we briefly recall two equivalent definitions of $D_L$, as well as the definition of a certain related (inequivalent) hardness measure, which we will need in order to prove our various lower bounds in their full generality.

For some DFA (one-way deterministic finite automaton) $M$, let $\lvert M \rvert$ denote the number of states of $M$ and let $L(M)$ denote the language of $M$ (i.e., the set of strings accepted by $M$). The earliest definition of a hardness measure equivalent to Dwork-Stockmeyer nonregularity was given by Karp \cite{karp1967some}, who defined $A_L(n)=\min \{\lvert M \rvert:M \text{ is a DFA and } L(M) \cap \Sigma^{\leq n}=L \cap \Sigma^{\leq n}\}$ to be the minimum number of states of a DFA that agrees with $L$ on all strings of length at most $n$; Shallit and Breitbart use the term \textit{deterministic automaticity} to refer to $A_L$. For any language $L$, it is immediately obvious that $A_L(n) \geq D_L(n), \forall n$; somewhat less obviously, $A_L(n)=D_L(n), \forall n$ \cite{shallit1996automaticityI,kacneps1990minimal,karp1967some}, and so the notions of nonregularity and deterministic automaticity coincide.

Consider a language $L \subseteq \Sigma^*$ and two communicating parties: Alice, who knows some string $x \in \Sigma^*$, and Bob, who knows some string $y \in \Sigma^*$. Alice sends some message $A(x) \in \{0,1\}^*$ to Bob, after which Bob must be able to determine, using $A(x)$ and $y$, if the string $w=xy$ is in $L$. Let $C_L(n)$ denote the maximum, taken over all $x,y \in \Sigma^*$ such that $\lvert xy \rvert \leq n$, of the number of bits sent from Alice to Bob by the optimal such (deterministic one-way) protocol. This quantity, the \textit{one-way deterministic communication complexity} of testing membership in $L$, is related to the nonregularity of $L$; in particular, $C_L(n)=\log D_L(n), \forall n$ \cite{condon1998power}.

Lastly, we recall the definition of a related (but inequivalent) hardness measure used by Ibarra and Ravikumar \cite{ibarra1988sublogarithmic} in their study of non-uniform small-space DTMs (deterministic Turing machines). Let $\Sigma^{n}=\{w \in \Sigma^*:\lvert w \rvert = n\}$. We then consider 2DFA (two-way deterministic finite automata), and use the same notation as was used above for DFA. For a language $L$, define $A_{L,=}^{2DFA}(n)=\min \{\lvert M \rvert:M \text{ is a 2DFA and } L(M) \cap \Sigma^n=L \cap \Sigma^n\}$ to be the minimum number of states of a 2DFA that agrees with $L$ on all strings of length exactly $n$. Clearly, for any language $L$, $A_{L,=}^{2DFA}(n)\leq A_L(n)$, $\forall n$. They then defined $\mathsf{NUDSPACE}(O(S(n)))$ (non-uniform deterministic space $O(S(n))$) to be the class of languages $L$ such that $A_{L,=}^{2DFA}(n)=2^{O(S(n))}$. Note that $\mathsf{NUDSPACE}(O(S(n)))=\mathsf{DSPACE}(O(S(n)))/2^{O(S(n))}$, the class of languages recognizable by a DTM that, on any input $w$, uses space $O(S(\lvert w \rvert))$, and has access to an ``advice'' string $y_{\lvert w \rvert}$, which depends only on the length $\lvert w \rvert$ of the input and is itself of length $\lvert y_n \rvert=2^{O(S(n))}$. In particular, $\mathsf{L/poly}:=\mathsf{DSPACE}(O(\log n))/2^{O(\log n)}=\mathsf{NUDSPACE}(O(\log n))=\{L:A_{L,=}^{2DFA}(n)=n^{O(1)}\}$. 

\subsection{A 2QCFA Analogue of the Dwork-Stockmeyer Lemma}\label{sec:runningTimeLowerBound:dworkStockmeyer}

We now prove that an analogue of the Dwork-Stockmeyer lemma holds for 2QCFA. The main idea is as follows. Suppose the 2QCFA $N$ recognizes $L \subseteq \Sigma^*$, with two-sided bounded error $\epsilon$, in expected time at most $T(n)$. We show that, if $D_L(n)$ is ``large,'' then, for any $m \in \mathbb{N}$, we can find $x,x' \in \Sigma^{\leq n}$ such that $x \not \sim_{L,n} x'$ and the distance between the corresponding $m$-truncated transfer operators $N_{x,m}^{\Balancedhookarrowright}$ and $N_{x',m}^{\Balancedhookarrowright}$ is ``small.'' By definition, $\exists y \in \Sigma^{\leq n-\max(\lvert x \rvert,\lvert x' \rvert)}$, such that $xy \in L \Leftrightarrow x'y \not \in L$; note that $xy,x'y \in \Sigma^{\leq n}$. Without loss of generality, we assume $xy \in L$, and hence $x'y \not \in L$. We also show that, for $m$ sufficiently large, if the distance between $N_{x,m}^{\Balancedhookarrowright}$ and $N_{x',m}^{\Balancedhookarrowright}$ is ``small,'' then the behavior of $N$ on the partitioned inputs $xy$ and $x'y$ will be similar; in particular, if $T(n)$ is ``small,'' then $\Pr[N \text{ accepts } xy] \approx \Pr[N \text{ accepts } x'y]$. However, as $xy \in L$, we must have $\Pr[N \text{ accepts } xy] \geq 1- \epsilon$, and as $x'y \not \in L$, we must have $\Pr[N \text{ accepts } x'y] \leq \epsilon$, which is impossible. This contradiction allows us to establish a lower bound on $T(n)$ in terms of $D_L(n)$. In this section, we formalize this idea. 

For $p \in \mathbb{N}_{\geq 1}$, we define the \textit{Schatten} $p$-\textit{norm} $\lVert \cdot \rVert_p:\L(V) \rightarrow \mathbb{R}_{\geq 0}$, where $\lVert Z \rVert_p=(\Tr((Z^{\dagger}Z)^{\frac{p}{2}}))^{\frac{1}{p}}$, $\forall Z \in \L(V)$. We also use the term \textit{trace norm} to refer to the Schatten $1$-norm. We define the \textit{induced trace norm} $\lVert \cdot \rVert_1:\T(V,V') \rightarrow \mathbb{R}_{\geq 0}$, where $\lVert \Phi \rVert_1=\sup \{\lVert \Phi(Z) \rVert_1 : Z \in \L(V), \lVert Z \rVert_1 \leq 1\}$, for any $\Phi \in \T(V,V')$. For $Z,Z' \in \L(\mathbb{C}^Q \otimes \mathbb{C}^C)$, we use $\lVert Z-Z' \rVert_1$, the distance metric induced by the trace norm, to measure the distance between $Z$ and $Z'$. For $x,x' \in \Sigma^*$ and $m \in \mathbb{N}$, we use $\lVert N_{x,m}^{\Balancedhookarrowright}-N_{x',m}^{\Balancedhookarrowright} \rVert_1$, the distance metric induced by the induced trace norm, to measure the distance between $N_{x,m}^{\Balancedhookarrowright}$ and $N_{x',m}^{\Balancedhookarrowright}$. 

Suppose $N$ is run on two distinct partitioned inputs $xy$ and $x'y$, producing two distinct $m$-truncated crossing sequences, following Definition~\ref{def:2qcfa:truncCross}. We first show that if $\lVert N_{x,m}^{\Balancedhookarrowright}-N_{x',m}^{\Balancedhookarrowright} \rVert_1$ is ``small'', then these crossing sequences are similar.

\begin{lemma}\label{thm:2qcfa:distanceCrossingSeq}
	Consider a 2QCFA $N$ with quantum basis states $Q$, classical states $C$, and input alphabet $\Sigma$. For $x,x',y \in \Sigma^*$ and $m \in \mathbb{N}$, let $Z_1,Z_2,\ldots \in \widehat{\den}(\mathbb{C}^Q \otimes \mathbb{C}^C)$ (resp. $Z_1',Z_2',\ldots \in \widehat{\den}(\mathbb{C}^Q \otimes \mathbb{C}^C)$) denote the $m$-truncated crossing sequence obtained when $N$ is run on $xy$ (resp. $x'y$). Then $\lVert Z_i-Z_i' \rVert_1 \leq \lfloor \frac{i-1}{2} \rfloor \lVert N_{x,m}^{\Balancedhookarrowright}-N_{x',m}^{\Balancedhookarrowright} \rVert_1$,  $\forall i \in \mathbb{N}_{\geq 1}$.	
\end{lemma}
\begin{proof}
	By definition, $Z_1=\ket{q_{\text{start}}}\bra{q_{\text{start}}} \otimes \ket{c_{\text{start}}}\bra{c_{\text{start}}}=Z_1'$, and so $\lVert Z_1-Z_1' \rVert_1=0$. Note that $\lVert \Phi(Z) \rVert_1 \leq \lVert Z \rVert_1$, $\forall Z \in \L(\mathbb{C}^Q \otimes \mathbb{C}^C)$, $\forall \Phi \in \chan(\mathbb{C}^Q \otimes \mathbb{C}^C)$ \cite[Corollary 3.40]{watrous2018theory}. Therefore, for any $\Phi \in \chan(\mathbb{C}^Q \otimes \mathbb{C}^C)$ and any $Z,Z' \in \L(\mathbb{C}^Q \otimes \mathbb{C}^C)$, we have $\lVert \Phi(Z)-\Phi(Z') \rVert_1=\lVert \Phi(Z-Z') \rVert_1 \leq \lVert Z-Z' \rVert_1$. By \Cref{thm:2qcfa:truncTranOpAndCrossProp}(\ref{thm:2qcfa:truncTranOpAndCrossProp:isChannel}), $N_{x,m}^{\Balancedhookarrowright},N_{x',m}^{\Balancedhookarrowright},N_{y,m}^{\Balancedhookarrowleft} \in \chan(\mathbb{C}^Q \otimes \mathbb{C}^C)$. For $i$ even, $Z_i=N_{y,m}^{\Balancedhookarrowleft}(Z_{i-1})$ and $Z_i'=N_{y,m}^{\Balancedhookarrowleft}(Z_{i-1}')$. We then have \[\lVert Z_i-Z_i' \rVert_1 =\lVert N_{y,m}^{\Balancedhookarrowleft}(Z_{i-1})-N_{y,m}^{\Balancedhookarrowleft}(Z_{i-1}') \rVert_1\leq \lVert Z_{i-1}-Z_{i-1}' \rVert_1.\] For odd $i>1$,  $Z_i=N_{x,m}^{\Balancedhookarrowright}(Z_{i-1})$ and $Z_i'=N_{x',m}^{\Balancedhookarrowright}(Z_{i-1}')$. We have $\lVert Z \rVert_1=1$, $\forall Z \in \den(\mathbb{C}^Q \otimes \mathbb{C}^C)$, which implies $\lVert \Phi(Z) \rVert_1 \leq \lVert \Phi \rVert_1$, $\forall \Phi \in \T(\mathbb{C}^Q \otimes \mathbb{C}^C)$. Therefore,
	\[\lVert Z_i-Z_i' \rVert_1=\lVert N_{x,m}^{\Balancedhookarrowright}(Z_{i-1})-N_{x',m}^{\Balancedhookarrowright}(Z_{i-1}') \rVert_1\] 
	\[\leq \lVert N_{x,m}^{\Balancedhookarrowright}(Z_{i-1})-N_{x,m}^{\Balancedhookarrowright}(Z_{i-1}') \rVert_1+\lVert N_{x,m}^{\Balancedhookarrowright}(Z_{i-1}')-N_{x',m}^{\Balancedhookarrowright}(Z_{i-1}') \rVert_1\]
	\[= \lVert N_{x,m}^{\Balancedhookarrowright}(Z_{i-1}-Z_{i-1}') \rVert_1+\lVert (N_{x,m}^{\Balancedhookarrowright}-N_{x',m}^{\Balancedhookarrowright})(Z_{i-1}') \rVert_1\leq \lVert Z_{i-1}-Z_{i-1}' \rVert_1+\lVert N_{x,m}^{\Balancedhookarrowright}-N_{x',m}^{\Balancedhookarrowright} \rVert_1\]
	The claim then follows by induction on $i \in \mathbb{N}_{\geq 1}$.
\end{proof}  

\begin{lemma}\label{thm:2qcfa:runTimeCrossSequenceDist}
	Consider a language $L \subseteq \Sigma^*$. Suppose $L \in  \mathsf{B2QCFA}(k,d,T(n),\epsilon)$, for some $k,d \in \mathbb{N}_{\geq 2}$, $T:\mathbb{N} \rightarrow \mathbb{N}$, and $\epsilon \in [0,\frac{1}{2})$. If, for some $n \in \mathbb{N}$, $\exists x,x' \in \Sigma^{\leq n}$ such that $x \not \sim_{L,n} x'$, then $T(n) \geq \frac{(1-2\epsilon)^2}{2} \lVert N_{x,m}^{\Balancedhookarrowright}-N_{x',m}^{\Balancedhookarrowright} \rVert_1^{-1}$, $\forall m \geq \lceil \frac{2}{1-2\epsilon} T(n) \rceil$.
\end{lemma}
\begin{proof}
	By definition, $x \not \sim_{L,n} x'$ precisely when $\exists y \in \Sigma^*$ such that $xy,x'y \in \Sigma^{\leq n}$, and $xy \in L \Leftrightarrow x'y \not \in L$. Fix such a $y$, and assume, without loss of generality, that $xy \in L$ (and hence $x'y \not \in L$). For $m \in \mathbb{N}$, suppose that, when $N$ is run on the partitioned input $xy$ (resp. $x'y$), we obtain the $m$-truncated crossing sequence $Z_{m,1},Z_{m,2},\ldots \in \widehat{\den}(\mathbb{C}^Q \otimes \mathbb{C}^C)$ (resp. $Z_{m,1}',Z_{m,2}',\ldots \in \widehat{\den}(\mathbb{C}^Q \otimes \mathbb{C}^C)$). For $c \in C$, let $E_c=\mathbbm{1}_{\mathbb{C}^Q} \otimes \ket{c}\bra{c} \in \L(\mathbb{C}^Q \otimes \mathbb{C}^C)$. For $s \in \mathbb{N}_{\geq 1}$, define $p_{m,s},p_{m,s}':C \rightarrow [0,1]$ such that $p_{m,s}(c)=\Tr(E_c Z_{m,s} E_c^{\dagger})$ and $p_{m,s}'(c)=\Tr(E_c Z_{m,s}' E_c^{\dagger})$. Then, for any $c \in C$, \Cref{thm:2qcfa:distanceCrossingSeq} implies \[\lvert p_{m,s}(c)-p_{m,s}'(c) \rvert =\lvert \Tr(E_c Z_{m,s} E_c^{\dagger})-\Tr(E_c Z_{m,s}' E_c^{\dagger}) \rvert=\lvert \Tr(E_c (Z_{m,s}-Z_{m,s}') E_c^{\dagger})\rvert\] \[\leq \lVert Z_{m,s} - Z_{m,s}' \rVert_1 \leq \frac{s-1}{2} \lVert N_{x,m}^{\Balancedhookarrowright}-N_{x',m}^{\Balancedhookarrowright} \rVert_1.\]		 
	
	Notice that $p_{m,s}(c_{\text{acc}})$ (resp. $p_{m,s}'(c_{\text{acc}})$) is the probability that $N$ accepts $xy$ (resp. $x'y$) within the first $s$ times (on a given branch of the computation) the head of $N$ crosses the boundary between $x$ (resp. $x'$) and $y$, where any branch that runs for more than $m$ steps between consecutive boundary crossings is forced to halt and reject immediately before attempting to perform the $m+1^{\text{st}}$ such step. Let $p_N(w)$ denote the probability that $N$ accepts an input $w \in \Sigma^*$, let $p_N(w,s)$ denote the probability that $N$ accepts $w$ within $s$ steps, and let $h_N(w,s)$ denote the probability that $N$ halts on input $w$ within $s$ steps. 
	
	Note that $x'y \not \in L$ implies $p_N(x'y) \leq \epsilon$. Clearly, $p_{m,s}'(c_{\text{acc}}) \leq p_N(x'y)$, for any $m$ and $s$, as all branches that attempt to perform more than $m$ steps (between consecutive crossings) are considered to reject the input in the $m$-truncated crossing sequence. Suppose $s \leq m$. Any branch that runs for a total of at most $s$ steps before halting is unaffected by $m$-truncation. Moreover, if a branch accepts within $s$ steps, it will certainly accept within $s$ crossings between $\#_L x$ and $y \#_R$. This implies $p_N(xy,s) \leq p_{m,s}(c_{\text{acc}})$.  Therefore, if $s \leq m$,  \[p_N(xy,s) \leq p_{m,s}(c_{\text{acc}}) \leq p_{m,s}'(c_{\text{acc}})+\lvert p_{m,s}(c_{\text{acc}})-p_{m,s}'(c_{\text{acc}})\rvert \leq \epsilon + \frac{s-1}{2} \lVert N_{x,m}^{\Balancedhookarrowright}-N_{x',m}^{\Balancedhookarrowright} \rVert_1.\]
	
	The expected running time of $N$ on input $xy$ is at most $T(\lvert xy \rvert)$. By Markov's inequality, $1-h_N(xy,s) \leq \frac{T(\lvert xy \rvert)}{s}$. Note that $xy \in L$ implies $p_N(xy) \geq 1-\epsilon$. Thus, for any $m \geq s \geq 1$,  \[1-\epsilon \leq p_N(xy) \leq p_N(xy,s)+(1-h_N(xy,s)) \leq \epsilon +  \frac{s-1}{2} \lVert N_{x,m}^{\Balancedhookarrowright}-N_{x',m}^{\Balancedhookarrowright} \rVert_1+\frac{T(\lvert xy \rvert)}{s}.\]  Set $s=\lceil \frac{2}{1-2\epsilon} T(n) \rceil$, and notice that $\lvert xy \rvert \leq n$ implies $T(\lvert xy \rvert) \leq T(n)$. For any $m \geq s$, \[1-2 \epsilon \leq \frac{\lceil \frac{2}{1-2\epsilon} T(n) \rceil-1}{2} \lVert N_{x,m}^{\Balancedhookarrowright}-N_{x',m}^{\Balancedhookarrowright} \rVert_1+\frac{T(\lvert xy \rvert)}{\lceil \frac{2}{1-2\epsilon} T(n) \rceil}\leq \frac{T(n)}{1- 2 \epsilon}\lVert N_{x,m}^{\Balancedhookarrowright}-N_{x',m}^{\Balancedhookarrowright} \rVert_1+\frac{1-2\epsilon}{2}.\] Therefore, $T(n) \geq \frac{(1-2\epsilon)^2}{2}\lVert N_{x,m}^{\Balancedhookarrowright}-N_{x',m}^{\Balancedhookarrowright} \rVert_1^{-1}, \ \  \forall m \geq \bigg\lceil \frac{2}{1-2\epsilon} T(n) \bigg\rceil$. \qedhere
\end{proof}

\begin{lemma}\label{thm:2qcfa:largeSetClosePair}
	Consider a 2QCFA $N=(Q,C,\Sigma,R,\theta,\delta,q_{\text{start}},c_{\text{start}},c_{\text{acc}},c_{\text{rej}})$. Let $k=\lvert Q \rvert$ and $d=\lvert C \rvert$. Consider any finite $X \subseteq \Sigma^*$ such that $\lvert X \rvert \geq 2$. Then $\forall m \in \mathbb{N}$, $\exists x,x' \in X$ such that $x \neq x'$ and $\lVert N_{x,m}^{\Balancedhookarrowright}-N_{x',m}^{\Balancedhookarrowright} \rVert_1 \leq 4\sqrt{2} k^4 d^2 \left(\lvert X \rvert^{\frac{1}{k^4 d^2}}-1 \right)^{-1}$. 
\end{lemma}
\begin{proof}
	For $q,q' \in Q$ and $c,c' \in C$, let $F_{q,q',c,c'}=\ket{q}\bra{q'} \otimes \ket{c}\bra{c'} \in \L(\mathbb{C}^Q \otimes \mathbb{C}^C )$. Let $J:\T(\mathbb{C}^Q \otimes \mathbb{C}^C)\rightarrow \L(\mathbb{C}^Q \otimes \mathbb{C}^C \otimes \mathbb{C}^Q \otimes \mathbb{C}^C)$ denote the \textit{Choi isomorphism}, which is given by $J(\Phi)=\sum_{(q,q',c,c') \in Q^2 \times C^2} F_{q,q',c,c'} \otimes \Phi(F_{q,q',c,c'}), \forall \Phi \in \T(\mathbb{C}^Q \otimes \mathbb{C}^C)$. Consider any $x \in \Sigma^*$ and $m \in \mathbb{N}$. We first show that, if $(c_1,c_2) \neq (c_1',c_2')$, then $\bra{q_2 c_2}  N_{x,m}^{\Balancedhookarrowright}(F_{q_1,q_1',c_1,c_1'}) \ket{q_2' c_2'}=0$. To see this, recall that, by Definition~\ref{def:2qcfa:truncTranOp}, $N_{x,m}^{\Balancedhookarrowright}=\Tr_{\mathbb{C}^{H_x}} \circ T_x \circ S_x^m \circ I_x$. If $c_1 \neq c_1'$, then $N_{x,m}^{\Balancedhookarrowright}(F_{q_1,q_1',c_1,c_1'})=\mathbbm{0}_{\mathbb{C}^Q \otimes \mathbb{C}^C}$, which implies $\bra{q_2 c_2}  N_{x,m}^{\Balancedhookarrowright}(F_{q_1,q_1',c_1,c_1'}) \ket{q_2' c_2'}=0$. If $c_2 \neq c_2'$, then $\bra{q_2 c_2}\Tr_{\mathbb{C}^{H_x}}(T_x(Z))\ket{q_2' c_2'}=0, \forall Z$, which implies $\bra{q_2 c_2}  N_{x,m}^{\Balancedhookarrowright}(F_{q_1,q_1',c_1,c_1'}) \ket{q_2' c_2'}=0$.
	
	Therefore, $\bra{q_2 c_2}  N_{x,m}^{\Balancedhookarrowright}(F_{q_1,q_1',c_1,c_1'}) \ket{q_2' c_2'}$ is only potentially non-zero at the $k^4 d^2$ elements where $(c_1,c_2)=(c_1',c_2')$. By \Cref{thm:2qcfa:truncTranOpAndCrossProp}(\ref{thm:2qcfa:truncTranOpAndCrossProp:isChannel}), $N_{x,m}^{\Balancedhookarrowright} \in \chan(\mathbb{C}^Q \otimes \mathbb{C}^C)$, which implies $J(N_{x,m}^{\Balancedhookarrowright}) \in \pos(\mathbb{C}^Q \otimes \mathbb{C}^C \otimes \mathbb{C}^Q \otimes \mathbb{C}^C)$ \cite[Corollary 2.27]{watrous2018theory}. Therefore, the elements where $(q_1,q_2) \neq (q_1',q_2')$ come in conjugate pairs, and the elements with $(q_1,q_2) \neq (q_1',q_2')$ are real. We define the function $g_{N,m}:\Sigma^* \rightarrow \mathbb{R}^{k^4 d^2}$ such that $g_{N,m}(x)$ encodes all the potentially non-zero $\bra{q_2 c_2}  N_{x,m}^{\Balancedhookarrowright}(F_{q_1,q_1',c_1,c_1'}) \ket{q_2' c_2'}$, without redundancy (only encoding one element of a conjugate pair). To be precise, the first $k^2 d^2$ entries of $g_{N,m}(x)$ are given by $\{\bra{q_2 c_2}  N_{x,m}^{\Balancedhookarrowright}(F_{q_1,q_1,c_1,c_1}) \ket{q_2 c_2}:q_1,q_2 \in Q,c_1,c_2 \in C\} \subseteq \mathbb{R}$. Establish some total order $\geq$ on $Q$, and let $\widehat{Q^4}=\{(q_1,q_1',q_2,q_2') \in Q^4:q_1' > q_1 \text{ or } (q_1'=q_1 \text{ and } q_2'>q_2)\}$. The remaining $k^4 d^2-k^2 d^2$ entries are given by encoding each of the $\frac{1}{2}(k^4 d^2-k^2 d^2)$ potentially non-zero entries $\{\bra{q_2 c_2}  N_{x,m}^{\Balancedhookarrowright}(F_{q_1,q_1',c_1,c_1}) \ket{q_2' c_2}:(q_1,q_1',q_2,q_2') \in \widehat{Q^4},c_1,c_2 \in C\} \subseteq \mathbb{C}$ as the pair of real numbers that comprise their real and imaginary parts.  
	
	Let $h=k^4 d^2$. Let $\lVert \cdot \rVert:\mathbb{R}^h \rightarrow \mathbb{R}_{\geq 0}$ denote the Euclidean $2$-norm and $\lVert \cdot \rVert_2:\L(V) \rightarrow \mathbb{R}_{\geq 0}$ denote the Schatten $2$-norm. Note that $\lVert \Phi \rVert_1 \leq \lVert J(\Phi) \rVert_1$, $\forall \Phi$ \cite[Section 3.4]{watrous2018theory}. We have, \[\lVert N_{x,m}^{\Balancedhookarrowright}-N_{x',m}^{\Balancedhookarrowright} \rVert_1 \leq \lVert J(N_{x,m}^{\Balancedhookarrowright}-N_{x',m}^{\Balancedhookarrowright}) \rVert_1 \leq \sqrt{\rank(J(N_{x,m}^{\Balancedhookarrowright}-N_{x',m}^{\Balancedhookarrowright}))}\lVert J(N_{x,m}^{\Balancedhookarrowright}-N_{x',m}^{\Balancedhookarrowright}) \rVert_2\] 
	\[\leq \sqrt{h} \lVert J(N_{x,m}^{\Balancedhookarrowright})-J(N_{x',m}^{\Balancedhookarrowright}) \rVert_2 \leq \sqrt{2h} \lVert g_{N,m}(x)-g_{N,m}(x') \rVert.\] 
	
	Note that $N_{x,m}^{\Balancedhookarrowright} \in \chan(\mathbb{C}^Q \otimes \mathbb{C}^C)$, which implies $\lVert N_{x,m}^{\Balancedhookarrowright}\rVert_1=1$ \cite[Corollary 3.40]{watrous2018theory}. Then, $\forall q,q' \in Q, \forall c \in C$, we have $\lVert F_{q,q',c,c} \rVert_1=1$, which implies $\lVert N_{x,m}^{\Balancedhookarrowright}(F_{q,q',c,c})\rVert_1\leq 1$. Therefore, \[\lVert g_{N,m}(x) \rVert \leq \lVert J(N_{x,m}^{\Balancedhookarrowright}) \rVert_2\leq \lVert J(N_{x,m}^{\Balancedhookarrowright}) \rVert_1 \leq \sum_{q,q' \in Q, c \in C} \lVert N_{x,m}^{\Balancedhookarrowright}(F_{q,q',c,c})\rVert_1 \leq k^2 d=\sqrt{h}.\]  
	
	For $v_0 \in \mathbb{R}^h$ and $r \in \mathbb{R}_{>0}$, let $B(v_0,r)=\{v \in \mathbb{R}^h: \lVert v_0-v \rVert \leq r\}$ denote the closed ball centered at $v_0$ of radius $r$ in $\mathbb{R}^h$, which has volume $\text{vol}(B(v_0,r))=c_h r^h$, for some constant $c_h \in \mathbb{R}_{>0}$. By the above, $\lVert g_{N,m}(x) \rVert \leq \sqrt{h}$, which implies that $B(g_{N,m}(x),\delta) \subseteq B(0,\sqrt{h}+\delta)$, $\forall \delta \in \mathbb{R}_{>0}$. Suppose $\forall x,x' \in X$ with $x \neq x'$, we have $B(g_{N,m}(x),\delta) \cap B(g_{N,m}(x'),\delta)=\emptyset$. Then $\sqcup_{x \in X} B(g_{N,m}(x),\delta) \subseteq B(0,\sqrt{h}+\delta)$, which implies $\lvert X \rvert c_h \delta^h \leq c_h (\sqrt{h}+\delta)^h$. Set $\delta=\frac{2\sqrt{h}}{\lvert X \rvert^{1/h}-1}$. Then $\exists x,x' \in X$, with $x \neq x'$, such that $B(g_{N,m}(x),\delta) \cap B(g_{N,m}(x'),\delta) \neq \emptyset$, which implies $\lVert g_{N,m}(x)-g_{N,m}(x') \rVert \leq 2 \delta$. Therefore, \[\lVert N_{x,m}^{\Balancedhookarrowright}-N_{x',m}^{\Balancedhookarrowright} \rVert_1 \leq \sqrt{2h} \lVert g_{N,m}(x)-g_{N,m}(x') \rVert \leq \sqrt{2h} 2\delta \leq 4\sqrt{2} k^4 d^2 \left(\lvert X \rvert^{\frac{1}{k^4 d^2}}-1 \right)^{-1}.\qedhere\]
\end{proof}  

We now prove a 2QCFA analogue of the Dwork-Stockmeyer lemma. 

\begin{theorem}\label{thm:2qcfa:lowerBoundRunningTime}
	If $L \in \mathsf{B2QCFA}(k,d,T(n),\epsilon)$, for some $k,d \in \mathbb{N}_{\geq 2}$, $T:\mathbb{N} \rightarrow \mathbb{N}$, and $\epsilon \in [0,\frac{1}{2})$, then $\exists N_0 \in \mathbb{N}$ such that $T(n) \geq \frac{(1-2\epsilon)^2}{16 \sqrt{2} k^4 d^2} D_L(n)^{\frac{1}{k^4 d^2}}$, $\forall n \geq N_0$.
\end{theorem}

\begin{proof}	
	Consider some $L \subseteq \Sigma^*$. By \cite[Lemma 3.1]{dwork1990time}, $L \in \mathsf{REG} \Leftrightarrow \exists b \in \mathbb{N}_{\geq 1}$ such that $D_L(n) \leq b$, $\forall n \in \mathbb{N}$. Thus, if $L \in \mathsf{REG}$, the claim is immediate (recall that $T(n)\geq n$). Next, suppose $L \not \in \mathsf{REG}$. For $n \in \mathbb{N}$, define $X_n=\{x_1,\cdots,x_{D_L(n)}\} \subseteq \Sigma^{\leq n}$ such that the $x_i$ are pairwise $(L,n)$-dissimilar. As $D_L(n)$ is not bounded above by any constant, $\exists N_0 \in \mathbb{N}$ such that $D_L(N_0) \geq 2^{k^4 d^2}$. Then, $\forall n \geq N_0$, we have $\lvert X_n \rvert =D_L(n) \geq D_L(N_0) \geq 2^{k^4 d^2}$. Fix $n \geq N_0$ and set $m=\lceil \frac{1- 2 \epsilon}{2} T(n)\rceil$. By \Cref{thm:2qcfa:largeSetClosePair}, $\exists x,x' \in X_n$ such that $x \neq x'$ and \[\lVert N_{x,m}^{\Balancedhookarrowright}-N_{x',m}^{\Balancedhookarrowright} \rVert_1 \leq 4\sqrt{2} k^4 d^2 \left(\lvert X_n \rvert^{\frac{1}{k^4 d^2}}-1 \right)^{-1} \leq 8\sqrt{2} k^4 d^2\lvert X_n \rvert^{-\frac{1}{k^4 d^2}}=8\sqrt{2} k^4 d^2 D_L(n)^{-\frac{1}{k^4 d^2}}.\] Fix such a pair $x,x'$, and note that $x \not \sim_{L,n} x'$, by construction. By \Cref{thm:2qcfa:runTimeCrossSequenceDist}, \[T(n) \geq  \frac{(1-2\epsilon)^2}{2}\lVert N_{x,m}^{\Balancedhookarrowright}-N_{x',m}^{\Balancedhookarrowright} \rVert_1^{-1} \geq \frac{(1-2\epsilon)^2}{16 \sqrt{2} k^4 d^2} D_L(n)^{\frac{1}{k^4 d^2}}.\qedhere\] 
\end{proof}

\subsection{2QCFA Running Time Lower Bounds and Complexity Class Separations}\label{sec:runningTimeLowerBound:2qcfaCompClasses}

Let $\mathsf{B2QCFA}(T(n))=\cup_{k,d \in \mathbb{N}_{\geq 2},\epsilon \in [0,\frac{1}{2})} \mathsf{B2QCFA}(k,d,T(n),\epsilon)$ denote the class of languages recognizable with two-sided bounded error by a 2QCFA with any constant number of quantum and classical states, in expected time at most $T(n)$. For a family $\mathcal{T}$ of functions of the form $T: \mathbb{N} \rightarrow \mathbb{N}$, let $\mathsf{B2QCFA}(\mathcal{T})=\cup_{T \in \mathcal{T}} \mathsf{B2QCFA}(T(n))$. We then write, for example, $\mathsf{B2QCFA}(2^{o(n)})$ to denote the union, taken over every function $T: \mathbb{N} \rightarrow \mathbb{N}$ such that $T(n)=2^{o(n)}$, of $\mathsf{B2QCFA}(T(n))$. We immediately obtain the following corollaries of \Cref{thm:2qcfa:lowerBoundRunningTime}.

\begin{corollary}\label{thm:2qcfa:coarseLowerBoundRunningTime}
	If $L \in \mathsf{B2QCFA}(T(n))$, then $D_L(n)=T(n)^{O(1)}$ and $C_L(n)=O(\log T(n))$.
\end{corollary}

\begin{corollary}\label{thm:2qcfa:expNonregularity}
	If a language $L$ satisfies $D_L(n) =2^{\Omega(n)}$, then $L \not \in \mathsf{B2QCFA}(2^{o(n)})$. 
\end{corollary}

Notice that $D_L(n)=2^{O(n)}$, for any $L$. We next exhibit a language for which $D_L(n) =2^{\Omega(n)}$, thereby yielding a strong lower bound on the running time of any 2QCFA that recognizes $L$. For $w=w_1 \cdots w_n \in \Sigma^*$, let $w^{\text{rev}}=w_n \cdots w_1$ denote the reversal of the string $w$. Let $L_{pal}=\{w \in \{a,b\}^*:w=w^{\text{rev}}\}$ consist of all palindromes over the alphabet $\{a,b\}$.

\begin{corollary}\label{thm:2qcfa:palindromeLowerBoundRunningTime}
	$L_{pal} \not \in \mathsf{B2QCFA}(2^{o(n)})$.
\end{corollary} 

\begin{proof}
	For $n \in \mathbb{N}$, let $W_n=\{w\in \{a,b\}^*: \lvert w \rvert = n\}$ denote all words over the alphabet $\{a,b\}$ of length $n$. For any $w,w' \in W_n$, with $w \neq w'$, we have $\lvert w w^{\text{rev}} \rvert=2n=\lvert w' w^{\text{rev}} \rvert$, $w w^{\text{rev}} \in L_{pal}$, and $w' w^{\text{rev}} \not \in L_{pal}$; therefore, $w \not \sim_{L_{pal},2n} w'$, $\forall w,w' \in W_n$ such that $w \neq w'$. This implies that $D_{L_{pal}}(2n) \geq \lvert W_n \rvert =2^n$. \Cref{thm:2qcfa:expNonregularity} then implies $L_{pal} \not \in \mathsf{B2QCFA}(T(n))$.
\end{proof} 

We define $\mathsf{BQE2QCFA}=\mathsf{B2QCFA}(2^{O(n)})$ to be the class of languages recognizable with two-sided bounded error in expected exponential time (with linear exponent) by a 2QCFA. Next, we say that a 2QCFA $N$ recognizes a language $L$ with \textit{negative one-sided bounded error} $\epsilon \in \mathbb{R}_{>0}$ if, $\forall w \in L$, $\Pr[N \text{ accepts } w] =1$, and, $\forall w \not \in L$, $\Pr[N \text{ accepts } w] \leq \epsilon$. We define $\mathsf{coR2QCFA}(k,d,T(n),\epsilon)$ as the class of languages recognizable with negative one-sided bounded error $\epsilon$ by a 2QCFA, with at most $k$ quantum basis states and at most $d$ classical states, that has expected running time at most $T(n)$ on all inputs of length at most $n$. We define $\mathsf{coR2QCFA}(T(n))$ and $\mathsf{coRQE2QCFA}$ analogously to the two-sided bounded error case.

Ambainis and Watrous \cite{ambainis2002two} showed that $L_{pal} \in \mathsf{coRQE2QCFA}$; in fact, their 2QCFA recognizer for $L_{pal}$ has only a single-qubit. Clearly, $\mathsf{coR2QCFA}(T(n)) \subseteq \mathsf{B2QCFA}(T(n))$, for any $T$, and $\mathsf{coRQE2QCFA} \subseteq \mathsf{BQE2QCFA}$. Therefore, the class of languages recognizable by a 2QCFA with bounded error in \textit{subexponential} time is properly contained in the class of languages recognizable by a 2QCFA in \textit{exponential} time.  

\begin{corollary}\label{thm:2qcfa:complexitySep}
	$\mathsf{B2QCFA}(2^{o(n)}) \subsetneq \mathsf{BQE2QCFA}$ and $\mathsf{coR2QCFA}(2^{o(n)}) \subsetneq \mathsf{coRQE2QCFA}$.		
\end{corollary}

We next define $\mathsf{BQP2QCFA}=\mathsf{B2QCFA}(n^{O(1)})$ to be the class of languages recognizable with two-sided bounded error in expected polynomial time by a 2QCFA. 

\begin{corollary}\label{thm:2qcfa:polyTime}
	If $L \in \mathsf{BQP2QCFA}$, then $D_L(n) =n^{O(1)}$. Therefore, $\mathsf{BQP2QCFA} \subseteq \mathsf{L/poly}$.
\end{corollary}
\begin{proof}
	The first statement is a special case of \Cref{thm:2qcfa:coarseLowerBoundRunningTime}. To see that $\mathsf{BQP2QCFA} \subseteq \mathsf{L/poly}$, recall that, as noted in \Cref{sec:runningTimeLowerBound:nonregularity}, $\mathsf{L/poly}=\{L:A_{L,=}^{2DFA}(n)=n^{O(1)}\}$; clearly, for any $L$ and any $n \in \mathbb{N}$, $A_{L,=}^{2DFA}(n) \leq A_L(n)=D_L(n)$.
\end{proof}

Of course, there are many languages $L$ for which one can establish a strong lower bound on $D_L(n)$, and thereby establish a strong lower bound on the expected running time $T(n)$ of any 2QCFA that recognizes $L$. In Section~\ref{sec:ComplexityWordProb}, we consider the case in which $L$ is the word problem of a group, and we show that very strong lower bounds can be established on $D_L(n)$. In the current section, we consider two especially interesting languages; the relevance of these languages was brought to our attention by Richard Lipton (personal communication). For $p \in \mathbb{N}$, let $\langle p \rangle_2 \in \{0,1\}^*$ denote its binary representation; let $L_{primes}=\{\langle p \rangle_2: p \text{ is prime}\}$. Note that $D_{L_{primes}}(n)=2^{\Omega(n)}$ \cite{shallit1996automaticityIV}, which immediately implies the following.

\begin{corollary}\label{thm:2qcfa:primeLowerBoundRunningTime}
	$L_{primes} \not \in \mathsf{B2QCFA}(2^{o(n)})$.
\end{corollary}

Say a string $w=w_1\cdots w_n \in \{0,1\}^n$ has a length-$3$ arithmetic progression (3AP) if $\exists i,j,k \in \mathbb{N}$ such that $1 \leq i <j <k \leq n$, $j-i=k-j$, and $w_i=w_j=w_k=1$; let $L_{3ap}=\{w\in \{0,1\}^*:w \text{ has a 3AP}\}$. It is straightforward to show the lower bound $D_{L_{3ap}}(n)=2^{n^{1-o(1)}}$, as well as the upper bound $D_{L_{3ap}}(n)=2^{n^{o(n)}}$. Therefore, one obtains the following lower bound on the running time of a 2QCFA that recognizes $L_{3ap}$, which, while still quite strong, is not as strong as that of $L_{pal}$ or $L_{primes}$.

\begin{corollary}\label{thm:2qcfa:3apLowerBoundRunningTime}
	$L_{3ap} \not \in \mathsf{B2QCFA}\left(2^{n^{1-\Omega(1)}}\right)$.
\end{corollary}

\begin{remark}
	While $L_{primes}$ and $L_{3ap}$ provide two more examples of natural languages for which our method yields strong lower bound on the running time of any 2QCFA recognizer, they also suggest the potential of proving a stronger lower bound for certain languages. That is to say, for $L_{pal}$, one has (essentially) matching lower and upper bounds on the running time of any 2QCFA recognizer; this is certainly not the case for $L_{primes}$ and $L_{3ap}$. In fact, we currently do not know if either $L_{primes}$ or $L_{3ap}$ can be recognized by a 2QCFA with bounded error \textit{at all} (i.e., regardless of time bound).
\end{remark}

\subsection{Transition Amplitudes of 2QCFA}\label{sec:runningTimeLowerBound:2qcfaTransitionAmplitudes}

As in Definition~\ref{def:2qcfa:transitionDesc}, for some 2QCFA $N=(Q,C,\Sigma,R,\theta,\delta,q_{\text{start}},c_{\text{start}},c_{\text{acc}},c_{\text{rej}})$, let $\{E_{c,\sigma,r,j}:r \in R, j \in J\} \subseteq \L(\mathbb{C}^Q)$ denote the set of operators that describe the selective quantum operation $\theta(c,\sigma)\in \quantop(\mathbb{C}^Q,R)$ that is applied to the quantum register when the classical state of $N$ is $c \in \widehat{C}$ and the head of $N$ is over the symbol $\sigma \in \Sigma_+$. The \textit{transition amplitudes} of $N$ are the set of numbers $\{\bra{q}E_{c,\sigma,r,j}\ket{q'}:c \in \widehat{C},\sigma \in \Sigma_+,r\in R, j \in J, q,q' \in Q\}\subseteq \mathbb{C}$.

While other types of finite automata are often defined without any restriction on their transition amplitudes, for 2QCFA, and other types of QFA, the allowed class of transition amplitudes strongly affects the power of the model. For example, using non-computable transition amplitudes, a 2QCFA can recognize certain undecidable languages with bounded error in expected polynomial time \cite{say2017magic}. Our lower bound holds even in this setting of unrestricted transition amplitudes. For $\mathbb{F} \subseteq \mathbb{C}$, we define complexity classes $\mathsf{coR2QCFA}_{\mathbb{F}}(k,d,T(n),\epsilon)$, $\mathsf{coRQE2QCFA}_{\mathbb{F}}$, etc., that are variants of the corresponding complexity class in which the 2QCFA are restricted to have transition amplitudes in $\mathbb{F}$. Using our terminology, Ambainis and Watrous \cite{ambainis2002two} showed that $L_{pal} \in \mathsf{coRQE2QCFA}_{\overline{\mathbb{Q}}}$, where $\overline{\mathbb{Q}}$ denotes the algebraic numbers, which are, arguably, the natural choice for the permitted class of transition amplitudes of a quantum model of computation. Therefore, $L_{pal}$ can be recognized with negative one-sided bounded error by a single-qubit 2QCFA with transition amplitudes that are all algebraic numbers in expected exponential time; however, $L_{pal}$ cannot be recognized with two-sided bounded error (and, therefore, not with one-sided bounded error) by a 2QCFA (of any constant size) in subexponential time, regardless of the permitted transition amplitudes. 

\section{Lower Bounds on the Running Time of Small-Space QTMs}\label{sec:sublogspaceQTM}

We next show that our technique also yields a lower bound on the expected running time of a quantum Turing machine (QTM) that uses sublogarithmic space (i.e., $o(\log n)$ space). The key idea is that a QTM $M$ that uses $S(n)$ space can be viewed as a sequence $(M_n)_{n \in \mathbb{N}}$ of 2QCFA, where $M_n$ has $2^{O(S(n))}$ (classical and quantum) states and $M_n$ simulates $M$ on all inputs of length at most $n$ (therefore, $M_n$ and $M$ have the same probability of acceptance and the same expected running time on any such input). The techniques of the previous section apply to 2QCFA with a sufficiently slowly growing number of states.

We consider the \textit{classically controlled} space-bounded QTM model that allows \textit{intermediate measurements}, following the definition of Watrous \cite{watrous2003complexity}. While several such QTM models have been defined, we focus on this model as we wish to prove our lower bound in the greatest generality possible. We note that the definitions of such QTM models by, for instance, Ta-Shma \cite{ta2013inverting}, Watrous \cite[Section VII.2]{watrous2009encyclopedia}, and (essentially, without the use of random access) van Melkebeek and Watson\cite{melkebeek2012time} are special cases of the QTM model that we consider. In the case of time-bounded quantum computation, it is well-known that allowing a QTM to perform intermediate measurements provably does not increase the power of the model; very recently, this fact has also been shown to hold in the simultaneously time-bounded and space-bounded setting \cite{fefferman2020eliminating}. 

A QTM has three tapes: (1) a classical read-only input tape, where each cell stores a symbol from the input alphabet (with special end-markers at the left and right ends), (2) a classical one-way infinite work tape, where each cell stores a symbol from some potentially larger (finite) alphabet, and (3) a one-way infinite quantum work tape, where each cell contains a single qubit. Each tape has a bidirectional (classical) head. A QTM also has a finite set of classical states that serve as its finite control, and a finite-size quantum register. 

The computation of a QTM is entirely \textit{classically controlled}. Each step of the computation consists of a \textit{quantum phase} followed by a \textit{classical phase}. In the quantum phase, depending on the current classical state and the symbols currently under the heads of the input tape and of the classical work tape, a QTM performs a selective quantum operation on the combined register consisting of its internal quantum register and the single qubit currently under the head of the quantum work tape. In the classical phase, depending on the current classical state, the symbols currently under the heads of the input tape and of the classical work tape, and the result of the operation performed in the quantum phase, a QTM updates its configuration as follows: a new classical state is entered, a symbol is written on the cell of the classical work tape under the head, and the heads of all tapes move at most one cell in either direction. 

A (branch of the computation of a) QTM halts and accepts/rejects its input by entering a special classical accept/reject state. As we wish to make our lower bound as strong as possible, we wish to be as generous as possible with the rejecting criteria of a QTM, and so we allow a QTM to also reject by looping (as we did with 2QCFA); similarly, no restriction is placed on the transition amplitudes of the QTM (see the discussion in Section~\ref{sec:runningTimeLowerBound:2qcfaTransitionAmplitudes}). Let $\mathsf{BQTISP}_{\epsilon}(T(n),S(n))$ denote the class of languages recognizable with two-sided bounded error $\epsilon \in [0,1/2)$ by a QTM that runs in at most $T(n)$ expected time, and uses at most $S(n)$ space, on all inputs of length at most $n$; of course, only the space used on the (classical and quantum) work tapes is counted. Furthermore, let $\mathsf{BQTISP}(T(n),S(n))=\cup_{\epsilon \in [0,1/2)} \mathsf{BQTISP}_{\epsilon}(T(n),S(n))$. 

As noted at the beginning of this section, we may view a QTM $M$ that operates in space $S(n)$ as a sequence of 2QCFA with a growing number of states. This yields the following analogue of Theorem~\ref{thm:2qcfa:lowerBoundRunningTime} for sublogarithmic-space QTMs.

\begin{theorem}\label{thm:qtm:lowerBoundRunningTime}
	Suppose $L \in \mathsf{BQTISP}(T(n),S(n))$, and suppose further that $S(n)=o(\log \log D_L(n))$. Then $\exists b_0 \in \mathbb{R}_{>0}$ such that, $T(n) =\Omega\big( 2^{-b_0 S(n)} D_L(n)^{2^{-b_0 S(n)}}\big)$.
\end{theorem}

\begin{proof}
	By definition, there is some QTM $M$ that recognizes $L$ with two-sided bounded error $\epsilon$, for some $\epsilon \in [0,1/2)$, where $M$ runs in expected time at most $T(n)$, and uses at most $S(n)$ space, on all inputs of length at most $n$. Let $F$ (resp. $P$) denote the finite set of classical states (resp. quantum basis states) of $M$, and let $\Sigma$ (resp. $\Gamma$) denote the finite input alphabet (resp. classical work tape alphabet) of $M$.
	
	For each $n \in \mathbb{N}$, we define a 2QCFA $M_n$ that correctly simulates $M$ on any $w \in \Sigma^{\leq n}$, in the obvious way. The (only) head of the 2QCFA $M_n$ (on its read-only input tape) directly simulates the head of the QTM $M$ on its read-only input tape. $M_n$ uses its classical states $C_n$ to keep track of the state $f \in F$ of the finite control of $M$, the string $y \in \Gamma^{S(n)}$ that appears in the first $S(n)$ cells of the classical work tape, and the positions $h_{c-work},h_{q-work} \in \{1,\ldots,S(n)\}$ of the heads on the (classical and quantum) work tapes. $M_n$ uses its quantum register, which has quantum basis states $Q_n$, to store the first $S(n)$ qubits of the quantum work tape and the $\log \lvert P \rvert$ qubits of the internal quantum register. The transition function of $M_n$ is defined such that, if $M_n$ is in a classic state $c \in C_n$ which (along with the head position on the input tape) completely specifies the classical part of a configuration of $M$, then $M_n$ performs the same quantum phase and classical phase that $M$ would in this configuration. Clearly, for any $w \in \Sigma^{\leq n}$, $M_n$ and $M$ have the same probability of acceptance and expected running time.
	
	Let $k_n=\lvert Q_n \rvert =\lvert P \rvert 2^{S(n)}$ denote the number of quantum basis states of $M_n$ and let $d_n=\lvert C_n \rvert =\lvert F \rvert \lvert \Gamma \rvert^{S(n)} S(n)^2$ denote the number of classical states of $M_n$. Then, $\exists b_0 \in \mathbb{R}_{>0}, \exists \widehat{N}_0 \in \mathbb{N}$ such that, $\forall n \geq \widehat{N}_0$, we have $k_n^4 d_n^2 \leq 2^{b_0 S(n)}$. Moreover, as $S(n)=o(\log \log D_L(n))$, $\exists \widetilde{N}_0 \in \mathbb{N}$ such that, $\forall n \geq \widetilde{N}_0$, $D_L(n)^{2^{-b_0 S(n)}} \geq 2$. Set $N_0=\max(\widehat{N}_0,\widetilde{N}_0)$. For any $n \geq N_0$, we may then construct $X_n \subseteq \Sigma^{\leq n}$ such that $\lvert X_n \rvert=D_L(n) \geq 2$ and the elements of $X_n$ are pairwise $(L,n)$-dissimilar. By Lemma~\ref{thm:2qcfa:largeSetClosePair}, $\exists x,x' \in X_n$ such that $x \neq x'$ and $$\lVert N_{x,m}^{\Balancedhookarrowright}-N_{x',m}^{\Balancedhookarrowright} \rVert_1 \leq 4\sqrt{2} k_n^4 d_n^2 \left(D_L(n)^{\frac{1}{k_n^4 d_n^2}}-1 \right)^{-1} \leq (4\sqrt{2})2^{b_0 S(n)} \left(D_L(n)^{2^{-b_0 S(n)}}-1 \right)^{-1}.$$   
	Let $a_{\epsilon}=\frac{(1-2\epsilon)^2}{2} \in \mathbb{R}_{>0}$. By Lemma~\ref{thm:2qcfa:runTimeCrossSequenceDist}, \[T(n) \geq a_{\epsilon}\lVert N_{x,m}^{\Balancedhookarrowright}-N_{x',m}^{\Balancedhookarrowright} \rVert_1^{-1} \geq \frac{a_{\epsilon}}{4 \sqrt{2}} 2^{-b_0 S(n)} \left(D_L(n)^{2^{-b_0 S(n)}}-1 \right)\geq \frac{a_{\epsilon}}{8 \sqrt{2}} 2^{-b_0 S(n)} D_L(n)^{2^{-b_0 S(n)}}.\qedhere\]
\end{proof}

\begin{remark}
	Recall that, for any language $L$, $D_L(n)=2^{O(n)}$; therefore, the supposition of the above theorem that $S(n)=o(\log \log D_L(n))$ implies $S(n)=o(\log n)$, and so this theorem only applies to QTMs that use sublogarithmic space. Moreover, this requirement also implies that $D_L(n)=\omega(1)$, and hence $L \not \in \mathsf{REG}$ \cite[Lemma 3.1]{dwork1990time}; of course, for any $L \in \mathsf{REG}$, we trivially have $L \in \mathsf{BQTISP}(n,O(1))$. 
\end{remark}

Note that, if $S(n)=o(\log n)$, then for any constants $b_1,b_2 \in \mathbb{R}_{>0}$, $2^{-b_1 S(n)} \geq n^{-b_2}$, for all sufficiently large $n$. We therefore obtain the following corollary.

\begin{corollary}\label{thm:qtm:coarseLowerBoundRunningTime}
	If $D_L(n) =2^{\Omega(n)}$, then $L \not \in \mathsf{BQTISP}\left(2^{n^{1-\Omega(1)}},o(\log n)\right)$. In particular, as $D_{L_{pal}}(n) =2^{\Omega(n)}$, $L_{pal} \not \in \ \mathsf{BQTISP}\left(2^{n^{1-\Omega(1)}},o(\log n)\right)$. 
\end{corollary}

\begin{remark}
	Of course, $L_{pal}$ can be recognized by a \textit{deterministic} TM in $O(\log n)$ space (and, trivially, polynomial time). Therefore, the previous corollary exhibits a natural problem for which polynomial time \textit{quantum} TM cannot outperform polynomial time \textit{deterministic} TM in terms of the amount of space used.
\end{remark}

\section{The Word Problem of a Group}\label{sec:ComplexityWordProb}

We begin by formally defining the word problem of a group; for further background, see, for instance \cite{loh2017geometric}. For a set $S$, let $F(S)$ denote the free group on $S$. For sets $S,R$ such that $R \subseteq F(S)$, let $N$ denote the normal closure of $R$ in $F(S)$; for a group $G$, if $G \cong F(S)/N$, then we say that $G$ \textit{has presentation} $\langle S|R \rangle$, which we denote by writing $G=\langle S|R \rangle$. Suppose $G=\langle S|R \rangle$, with $S$ finite; we now define $W_{G=\langle S|R \rangle}$, \textit{the word problem of} $G$ \textit{with respect to the presentation} $\langle S|R \rangle$. We define the set of formal inverses $S^{-1}$, such that, for each $s \in S$, there is a unique corresponding $s^{-1} \in S^{-1}$, and $S \cap S^{-1}= \emptyset$. Let $\Sigma=S \sqcup S^{-1}$, let $\Sigma^*$ denote the free monoid over $\Sigma$, and let $\phi:\Sigma^* \rightarrow G$ be the natural (monoid) homomorphism that takes each string in $\Sigma^*$ to the element of $G$ that it represents. We use $1_G$ to denote the identity element of $G$. Then $W_{G=\langle S|R \rangle}=\phi^{-1}(1_G)$. Note that the definition of the word problem does depend on the choice presentation. However, if $\mathcal{L}$ is any complexity class that is closed under inverse homomorphism, then if $\langle S|R \rangle$ and $\langle S'|R' \rangle$ are both presentations of some group $G$, and $S$ and $S'$ are both finite, then $W_{G=\langle S|R \rangle} \in \mathcal{L} \Leftrightarrow W_{G=\langle S'|R' \rangle} \in \mathcal{L}$ \cite{herbst1991subclass}. As all complexity classes considered in this paper are easily seen to be closed under inverse homomorphism, we will simply write $W_G \in \mathcal{L}$ to mean that $W_{G=\langle S|R \rangle} \in \mathcal{L}$, for every presentation $G=\langle S|R \rangle$, with $S$ finite. We note that the languages $L_{pal}$ and $L_{eq}$, which Ambainis and Watrous \cite{ambainis2002two} showed satisfy $L_{pal} \in \mathsf{coRQE2QCFA}_{\overline{\mathbb{Q}}}$ and $L_{eq} \in \mathsf{BQP2QCFA}$, are closely related to the word problems of the groups $F_2$ and $\mathbb{Z}$, respectively.

\subsection{The Growth Rate of a Group and Nonregularity}\label{sec:ComplexityWordProb:growthRate}

Consider a group $G = \langle S|R \rangle$, with $S$ finite. Define $\Sigma$ and $\phi$ as in the previous section. For $g \in G$, let $l_S(g)$ denote the smallest $m \in \mathbb{N}$ such that $\exists \sigma_1,\ldots,\sigma_m \in \Sigma$ such that $g=\phi(\sigma_1 \cdots \sigma_m)$. For $n \in \mathbb{N}$, we define $B_{G,S}(n)=\{g \in G:l_S(g) \leq n\}$ and we further define $\beta_{G,S}(n)=\lvert B_{G,S}(n) \rvert$, which we call the \textit{growth rate of} $G$ \textit{with respect to} $S$. The following straightforward lemma demonstrates an important relationship between $\beta_{G,S}$ and $D_{W_{G=\langle S | R \rangle}}$.

\begin{lemma}\label{thm:growthRateDistinguish}
	Suppose $G = \langle S|R \rangle$ with $S$ finite. Using the notation established above, let $W_G:=W_{G=\langle S |R \rangle}=\phi^{-1}(1_G)$ denote the word problem of $G$ with respect to this presentation. Then, $\forall n \in \mathbb{N}$, $D_{W_G}(2n) \geq \beta_{G,S}(n)$.
\end{lemma}
\begin{proof}
	Fix $n \in \mathbb{N}$, let $k=\beta_{G,S}(n)$, and let $B_{G,S}(n)=\{g_1,\ldots,g_k\}$. For a string $x=x_1 \cdots x_m \in \Sigma^*$, where each $x_j \in \Sigma$, let $\lvert x \rvert=m$ denote the (string) length of $x$ and define $x^{-1}=x_m^{-1} \cdots x_1^{-1}$. Note that, $\forall g \in G$, $l_S(g)=\min_{w \in \phi^{-1}(g)} \lvert w \rvert$. Therefore, for each $i \in \{1,\ldots,k\}$ we may define $w_i \in \phi^{-1}(g_i)$ such that $\lvert w_i \rvert =l_S(g_i)$. Observe that $w_i w_i^{-1} \in W_G$ and $\lvert w_i w_i^{-1}\lvert=2\lvert w_i \lvert=2l_S(g_i)\leq 2n$; moreover, for each $j \neq i$, we have  $w_j w_i^{-1} \not \in W_G$ and $\lvert w_j w_i^{-1} \rvert=\lvert w_j \rvert +\lvert w_i \rvert =l_S(g_j)+l_S(g_i)\leq 2n$. Therefore, $w_1,\ldots,w_k$ are pairwise $(W_G,2n)$-dissimilar, which implies $D_{W_G}(2n) \geq k =\beta_{G,S}(n)$.
\end{proof}

\begin{remark}
	In fact, one may also easily show that $D_{W_G}(2n) \leq \beta_{G,S}(n)+1$, though we do not need this here. Essentially, $ \beta_{G,S}(n)$ is (another) equivalent characterization of the nonregularity $D_{W_G}(2n)$ (see Section~\ref{sec:runningTimeLowerBound:nonregularity} for a discussion of the many such characterizations of nonregularity).
\end{remark}

For a pair of non-decreasing functions $f_1,f_2: \mathbb{R}_{\geq 0} \rightarrow \mathbb{R}_{\geq 0}$, we write $f_1 \prec f_2$ if $\exists C_1,C_2 \in \mathbb{R}_{>0}$ such that $\forall r \in \mathbb{R}_{\geq 0}$, $f_1(r) \leq C_1 f_2(C_1 r+C_2)+C_2$;  we write $f_1 \sim f_2$ if both $f_1 \prec f_2$ and $f_2 \prec f_1$. Suppose $\langle S|R \rangle$ and $\langle S'|R' \rangle$ are both presentations of $G$, with $S$ and $S'$ finite. It is straightforward to show that $\beta_{G,S}$ and $\beta_{G,S'}$ are non-decreasing, and that $\beta_{G,S}\sim \beta_{G,S'}$ \cite[Proposition 6.2.4]{loh2017geometric}. For this reason, we will simply write $\beta_G$ to denote the growth rate of $G$. 

\begin{definition2}\label{def:growthTypes}
	Suppose $G$ is a finitely generated group. If $\beta_G \sim (n \mapsto e^n)$, we say $G$ \textit{has exponential growth}. If $\exists c \in \mathbb{R}_{\geq 0}$ such that $\beta_G \prec (n \mapsto n^c)$, we say $G$ \textit{has polynomial growth}. Otherwise, we say $G$ \textit{has intermediate growth}. Note that, for any finitely generated group $G$, we have $\beta_G \prec (n \mapsto e^n)$, and so the term ``intermediate growth'' is justified. 
\end{definition2}

\subsection{Word Problems Recognizable by 2QCFA and Small-Space QTMs}\label{sec:ComplexityWordProb:recBy2QCFAandSmallQTM}

By making use of two very powerful results in group theory, the Tits' Alternative \cite{tits1972free} and Gromov's theorem on groups of polynomial growth \cite{gromov1981groups}, we exhibit useful lower bounds on $D_{W_G}$, which in turn allows us to show a strong lower bound on the expected running time of a 2QCFA that recognizes $W_G$. 

\begin{theorem}\label{thm:2qcfa:wordProbRunTimeMain}
	For any finitely generated group $G$, the following statements hold.
	\begin{enumerate}[(i)]
		\item\label{thm:2qcfa:wordProbRunTimeMain:basic} If $W_G \in \mathsf{B2QCFA}(k,d,T(n),\epsilon)$, then $\beta_G \prec (n \mapsto T(n)^{k^4 d^2})$.
		\item\label{thm:2qcfa:wordProbRunTimeMain:exp} If $G$ has exponential growth, then $W_G \not \in \mathsf{B2QCFA}(2^{o(n)})$.
		\item\label{thm:2qcfa:wordProbRunTimeMain:linGroupNotVirtNilp} If $G$ is a linear group over a field of characteristic $0$, and $G$ is not virtually nilpotent, then $W_G \not \in \mathsf{B2QCFA}(2^{o(n)})$.
		\item\label{thm:2qcfa:wordProbRunTimeMain:poly} If $W_G \in \mathsf{BQP2QCFA}$, then $G$ is virtually nilpotent.
	\end{enumerate}
\end{theorem}
\begin{proof}
	\begin{enumerate}[(i)]
		\item Follows immediately from Lemma~\ref{thm:growthRateDistinguish} and Corollary~\ref{thm:2qcfa:coarseLowerBoundRunningTime}.
		\item Follows immediately from Definition~\ref{def:growthTypes} and part (i) of this theorem.
		\item As a consequence of the famous Tits' Alternative \cite{tits1972free}, every finitely generated linear group over a field of characteristic $0$ either has polynomial growth or exponential growth, and has polynomial growth precisely when it is virtually nilpotent (\cite[Corollary 1]{tits1972free},\cite{wolf1968growth}). The claim then follows by part (ii) of this theorem.
		\item If $W_G \in \mathsf{BQP2QCFA}$, then $W_G \in \mathsf{B2QCFA}(k,d,n^c,\epsilon)$ for some $k,d,c \in \mathbb{N}_{\geq 1},\epsilon \in [0,\frac{1}{2})$. By part (i) of this theorem, $\beta_G \prec (n \mapsto n^{ck^4 d^2})$, which implies $G$ has polynomial growth. By Gromov's theorem on groups of polynomial growth \cite{gromov1981groups}, a finitely generated group has polynomial growth precisely when it is virtually nilpotent. \qedhere
	\end{enumerate}
\end{proof}

\begin{remark}
	All \textit{known} $G$ of intermediate growth have $\beta_G \sim (n \mapsto e^{n^c})$, for some $c \in (1/2,1)$. Therefore, a strong lower bound may be established on the running time of any 2QCFA that recognizes $W_G$, for any known group of intermediate growth. We also note that one may show that the conclusion of \Cref{thm:2qcfa:wordProbRunTimeMain}(\ref{thm:2qcfa:wordProbRunTimeMain:poly}) still holds even if $W_G$ is only assumed to be recognized in slightly super-polynomial time. In particular, by a quantitative version of Gromov's theorem due to Shalom and Tal \cite[Corollary 1.10]{shalom2010finitary}, $\exists c \in \mathbb{R}_{>0}$ such that if $\beta_{G,S}(n) \leq n^{c(\log \log n)^c}$, for some $n>1/c$, then $G$ is virtually nilpotent. 
\end{remark}

Let $\mathcal{G}_{\textbf{vAb}}$ (resp. $\mathcal{G}_{\textbf{vNilp}}$) denote the collection of all finitely generated virtually abelian (resp. nilpotent) groups. Let $\U(k,\overline{\mathbb{Q}})$ denote the group of $k \times k$ unitary matrices with algebraic number entries, and let $\mathcal{U}$ consist of all finitely generated subgroups of any $\U(k,\overline{\mathbb{Q}})$. We have recently shown that if $G \in \mathcal{U}$, then $W_G \in \mathsf{coRQE2QCFA}_{\overline{\mathbb{Q}}}$ \cite[Corollary 1.4.1]{remscrim2019power}. Observe that $\mathcal{G}_{\textbf{vAb}} \subseteq \mathcal{U}$ and that all groups in $\mathcal{U}$ are finitely generated linear groups over a field of characteristic zero. Moreover, $\mathcal{U} \cap \mathcal{G}_{\textbf{vNilp}}= \mathcal{G}_{\textbf{vAb}}$ \cite[Proposition 2.2]{thom2013convergent}. We, therefore, obtain the following corollary of \Cref{thm:2qcfa:wordProbRunTimeMain}, which exhibits a broad and natural class of languages that a 2QCFA can recognize in exponential time, but not in subexponential time. 

\begin{corollary}\label{thm:2qcfa:wordProbExp}
	$\forall G \in \mathcal{U} \setminus \mathcal{G}_{\textbf{vAb}}$, we have $W_G \in \mathsf{coRQE2QCFA}_{\overline{\mathbb{Q}}}$ but $W_G \not \in \mathsf{B2QCFA}(2^{o(n)})$.
\end{corollary}

We have also recently shown that $W_G \in \mathsf{coRQP2QCFA}_{\overline{\mathbb{Q}}}(2) \subseteq \mathsf{BQP2QCFA}$, $\forall G \in \mathcal{G}_{\textbf{vAb}}$ \cite[Theorem 1.2]{remscrim2019power}. By \Cref{thm:2qcfa:wordProbRunTimeMain}, if $W_G \in \mathsf{BQP2QCFA}$, then $G \in \mathcal{G}_{\textbf{vNilp}}$. This naturally raises the question of whether or not there is some $G \in \mathcal{G}_{\textbf{vNilp}} \setminus \mathcal{G}_{\textbf{vAb}}$ such that $W_G \in \mathsf{BQP2QCFA}$. Consider the (three-dimensional discrete) Heisenberg group $H=\langle x,y,z|z=[x,y],[x,z]=[y,z]=1 \rangle$. $W_H$ is a natural choice for a potential ``hard'' word problem for 2QCFA, due to the lack of faithful finite-dimensional unitary representations of $H$ (see \cite{remscrim2019power} for further discussion). We next show that if $W_H  \not \in \mathsf{BQP2QCFA}$, then we have a complete classification of those word problems recognizable by 2QCFA in polynomial time.

\begin{proposition}\label{thm:2qcfa:wordProbHeisenberg}
	If $W_H  \not \in \mathsf{BQP2QCFA}$, then $W_G \in \mathsf{BQP2QCFA} \Leftrightarrow G \in \mathcal{G}_{\textbf{vAb}}$.
\end{proposition}
\begin{proof}
	By the above discussion, it suffices to show the following claim: if $W_G \in \mathsf{BQP2QCFA}$, for some $G \in \mathcal{G}_{\textbf{vNilp}} \setminus \mathcal{G}_{\textbf{vAb}}$, then $W_H  \in \mathsf{BQP2QCFA}$. Begin by noting that $\forall G \in \mathcal{G}_{\textbf{vNilp}} \setminus \mathcal{G}_{\textbf{vAb}}$, $G$ has a subgroup isomorphic to $H$ \cite[Theorem 12]{holt2005groups}. It is straightforward to see that $\mathsf{BQP2QCFA}$ is closed under inverse homomorphism and intersection with regular languages. Therefore, if $W_G \in \mathsf{BQP2QCFA}$, then $W_{H} \in \mathsf{BQP2QCFA}$ \cite[Lemma 2]{holt2005groups}. 
\end{proof}

We next obtain the following analogue of Theorem~\ref{thm:2qcfa:wordProbRunTimeMain} for small-space QTMs.

\begin{theorem}\label{thm:qtm:wordProbRunTimeMain}
	For any finitely generated group $G$, the following statements hold.
	\begin{enumerate}[(i)]
		\item\label{thm:qtm:wordProbRunTimeMain:exp} If $G$ has exponential growth, then $W_G \not \in \mathsf{BQTISP}(2^{n^{1-\Omega(1)}},o(\log n))$.
		\item\label{thm:qtm:wordProbRunTimeMain:linGroupNotVirtNilp} If $G$ is a linear group over a field of characteristic $0$, and $G$ is not virtually nilpotent, then $W_G \not \in \mathsf{BQTISP}(2^{n^{1-\Omega(1)}},o(\log n))$.
		\item\label{thm:qtm:wordProbRunTimeMain:poly} If $W_G \in \mathsf{BQTISP}(n^{O(1)},o(\log \log \log n))$, then $G$ is virtually nilpotent.
	\end{enumerate}
\end{theorem}
\begin{proof}
	\begin{enumerate}[(i)]
		\item Follows immediately from Corollary~\ref{thm:qtm:coarseLowerBoundRunningTime} and Lemma~\ref{thm:growthRateDistinguish}.
		\item The claim follows from the Tits' Alternative \cite{tits1972free} and the first part of this theorem.
		\item If $W_G \in \mathsf{BQTISP}(n^{O(1)},o(\log \log \log n))$, then $\forall c \in \mathbb{R}_{>0}$ and for all sufficiently large $n$ we have, by Theorem~\ref{thm:qtm:lowerBoundRunningTime}, $D_L(n) \leq n^{c(\log \log n)^c}$. By Lemma~\ref{thm:growthRateDistinguish} and the quantitative version of Gromov's theorem due to Shalom and Tal \cite[Corollary 1.10]{shalom2010finitary}, $G$ is virtually nilpotent. \qedhere
	\end{enumerate}
\end{proof}

\section{Discussion}\label{sec:discussion}

In this paper, we established strong lower bounds on the expected running time of 2QCFA, or sublogarithmic-space QTMs, that recognize particular languages with bounded error. In particular, the language $L_{pal}$ had been shown by Ambainis and Watrous \cite{ambainis2002two} to be recognizable with bounded error by a single-qubit 2QCFA in expected time $2^{O(n)}$. We have given a matching lower bound: no 2QCFA (of any size) can recognize $L_{pal}$ with bounded error in expected time $2^{o(n)}$. Moreover, we have shown that no QTM, that runs in expected time $2^{n^{1-\Omega(1)}}$ and uses space $o(\log n)$, can recognize $L_{pal}$ with bounded error. This latter results is especially interesting, as a \textit{deterministic} TM can recognize $L_{pal}$ using space $O(\log n)$ (and, of course, polynomial time); therefore, polynomial time \textit{quantum} TMs have no (asymptotic) advantage over polynomial time \textit{deterministic} TMs in terms of the amount of space needed to recognize $L_{pal}$.

Our main technical result, Theorem~\ref{thm:2qcfa:lowerBoundRunningTime}, showed that, if a language $L$ is recognized with bounded error by a 2QCFA in expected time $T(n)$, then $\exists a \in \mathbb{R}_{>0}$ (that depends only on the number of states of the 2QCFA) such that $T(n)=\Omega(D_L(n)^a)$, where $D_L$ is the Dwork-Stockmeyer \textit{nonregularity} of $L$. This result is extremely (qualitatively) similar to the landmark result of Dwork and Stockmeyer \cite[Lemma 4.3]{dwork1990time}, which showed that, if a language $L$ is recognized with bounded error by a 2PFA in expected time $T(n)$, then $\exists a \in \mathbb{R}_{>0}$ (that depends only on the number of states of the 2PFA) such that $T(n)=\Omega(2^{D_L(n)^a})$. We again note that both of these lower bounds are tight.

We conclude by stating a few interesting open problems. While our lower bound on the expected running time $T(n)$, of a 2QCFA that recognizes a language $L$, in terms of $D_L(n)$ cannot be improved, it is natural to ask if one could establish a lower bound on $T(n)$ in terms of a different hardness measure of $L$ that would be stronger for certain languages. Generalizing the definitions made in Section~\ref{sec:runningTimeLowerBound:nonregularity}, let $\mathcal{F}$ denote a class of finite automata (e.g., DFA, NFA, 2DFA, etc.), let $L$ be a language over some alphabet $\Sigma$, and let $A_{L,\leq}^{\mathcal{F}}(n)=\min \{\lvert M \rvert:M \in \mathcal{F} \text{ and } L(M)\cap \Sigma^{\leq n}=L \cap \Sigma^{\leq n}\}$ denote the smallest number of states of an automaton of type $\mathcal{F}$ that agrees with $L$ on all strings of length at most $n$. As discussed earlier, $A_{L,\leq}^{DFA}(n)=D_L(n)$, for any language $L$ and for any $n \in \mathbb{N}$. Recall that DFA and 2DFA both recognize precisely the regular languages \cite{rabin1959finite}, but for some $\widehat{L} \in \mathsf{REG}$, the smallest 2DFA that recognizes $\widehat{L}$ might have many fewer states than the smallest DFA that recognizes $\widehat{L}$. In fact, there is a sequence of regular languages $(L_k)_{k \in \mathbb{N}}$ such that $L_k$ can be recognized by a $5k+5$-state 2DFA, but any DFA that recognizes $L_k$ requires at least $k^k$ states \cite{meyer1971economy}; however, this is (essentially) the largest succinctness advantage possible, as any language recognizable by a $d$-state 2DFA is recognizable by a $(d+2)^{d+1}$-state DFA \cite{shepherdson1959reduction}. Of course, for any language $L$, we have $A_{L,\leq}^{2DFA}(n) \leq A_{L,\leq}^{DFA}(n)$, $\forall n$. For certain languages $L$, we have $A_{L,\leq}^{2DFA}(n) \ll A_{L,\leq}^{DFA}(n)$, $\forall n$; most significantly, this holds for the languages $L_{pal}$ and $L_{eq}$ shown by Ambainis and Watrous \cite{ambainis2002two} to be recognizable with bounded error by 2QCFA in, respectively, expected exponential time and expected polynomial time. In particular, it is easy to show that $A_{L_{pal},\leq}^{DFA}(n)=2^{\Theta(n)}$, $A_{L_{eq},\leq}^{DFA}(n)=\Theta(n)$, and $A_{L_{pal},\leq}^{2DFA}(n)=n^{\Theta(1)}$; moreover, $A_{L_{eq},\leq}^{2DFA}(n)=\log^{\Theta(1)}(n)$ \cite[Theorem 3 and Corollary 4]{ibarra1988sublogarithmic}. In fact, this same phenomenon occurs for all the group word problems that we can show \cite{remscrim2019power} are recognized by 2QCFA. Might this be true for all languages recognizable by 2QCFA?

\begin{openProb}
	If a language $L$ is recognizable with bounded error by a 2QCFA in expected time $T(n)$, does a stronger lower bound than $T(n)=(A_{L,\leq}^{2DFA}(n))^{\Omega(1)}$ hold?
\end{openProb}

We have shown that the class of languages recognizable with bounded error by a 2QCFA in expected polynomial time is contained in $\mathsf{L/poly}$. This type of \textit{dequantization} result, which shows that the class of languages recognizable by a particular quantum model is contained in the class of languages recognizable by a particular classical model, is analogous to the Adleman-type \cite{adleman1978two} \textit{derandomization} result $\mathsf{BPL} \subseteq \mathsf{L/poly}$. It is natural to ask if our dequantization result might be extended, either to 2QCFA that run in a larger time bound, or to small-space QTM. Note that $\mathsf{L/poly}=\{L:A_{L,=}^{2DFA}(n)=n^{O(1)}\}=\{L:A_{L,\leq}^{2DFA}(n)=n^{O(1)}\} \supsetneq \{L:A_{L,\leq}^{DFA}(n)=n^{O(1)}\}$. This further demonstrates the value of the preceding open problem, as any improvement in the lower bound on $T(n)$ in terms of $A_{L,\leq}^{2DFA}(n)$ would directly translate into an improved dequantization result.

The seminal paper of Lipton and Zalcstein \cite{lipton1977word} showed that, if a finitely generated group $G$ has a faithful finite-dimensional (linear) representation over a field of characteristic $0$, then $W_G \in \mathsf{L}$ (deterministic logspace). We \cite{remscrim2019power} recently adapted their technique to show that 2QCFA can recognize the word problem $W_G$ of any group $G$ that belongs to a certain (proper) subset of the set of groups to which their result applies: any group $G$ that has a faithful finite-dimensional \textit{unitary} representation of a certain special type. The requirement, imposed by the laws of quantum mechanics, that the state of the quantum register of a 2QCFA must evolve unitarily, prevents a 2QCFA from (directly) implementing the Lipton-Zalcstein algorithm for any other groups; on the other hand, for those groups $G$ that do have such a representation, these same laws allow a 2QCFA to recognize $W_G$ using only a constant amount of space. The word problem $W_G$ of any group $G$ that lacks such a representation (for example, all $G \in \mathcal{G}_{\textbf{vNilp}} \setminus \mathcal{G}_{\textbf{vAb}}$, or any infinite Kazhdan group, or any group of intermediate growth) seems to be a plausible candidate for a hard problem for 2QCFA (see \cite{remscrim2019power} for further discussion).

\begin{openProb}\label{opProb:2qcfa:wordProbExp}
	Is there a finitely generated group $G$ that does not have a faithful finite-dimensional projective unitary representation for which $W_G \in \mathsf{BQE2QCFA}$? 
\end{openProb}

Concerning those groups with word problem recognizable by a 2QCFA in expected polynomial time, we have shown that, if $G \in \mathcal{G}_{\textbf{vAb}}$, then $W_G \in \mathsf{coRQP2QCFA}_{\overline{\mathbb{Q}}}(2) \subseteq \mathsf{BQP2QCFA}$ \cite[Theorem 1.2]{remscrim2019power}; moreover, if $W_G \in \mathsf{BQP2QCFA}$, then $G \in \mathcal{G}_{\textbf{vNilp}}$ (\Cref{thm:2qcfa:wordProbRunTimeMain}). We have also shown, if $W_H \not \in \mathsf{BQP2QCFA}$, where $H \in \mathcal{G}_{\textbf{vNilp}}$ is the (three-dimensional discrete) Heisenberg group, then the classification of those groups whose word problem is recognizable by a 2QCFA in expected polynomial time would be complete; in particular, we would have $W_G \in \mathsf{BQP2QCFA} \Leftrightarrow G \in \mathcal{G}_{\textbf{vAb}}$ (\Cref{thm:2qcfa:wordProbHeisenberg}). This naturally raises the following question.

\begin{openProb}\label{opProb:2qcfa:wordProbPoly}
	Is there a group $G \in \mathcal{G}_{\textbf{vNilp}} \setminus \mathcal{G}_{\textbf{vAb}}$ such that $W_G \in \mathsf{BQP2QCFA}$? In particular, is $W_H \in \mathsf{BQP2QCFA}$, where $H$ is the Heisenberg group?
\end{openProb}

\section*{Acknowledgments}

The author would like to express his sincere gratitude to Professor Michael Sipser for many years of mentorship and support, without which this work would not have been possible, and to thank Professor Richard Lipton, as well as the anonymous reviewers, for several helpful comments on an earlier draft of this paper.

\bibliographystyle{plainurl}
\let\OLDthebibliography\thebibliography
\renewcommand\thebibliography[1]{
	\OLDthebibliography{#1}
	\setlength{\parskip}{0pt}
	\setlength{\itemsep}{0pt plus 0.3ex}
}
\bibliography{qfaLinGroupRef} 

\begin{thebibliography}{10}

\bibitem{adleman1978two}
Leonard Adleman.
\newblock Two theorems on random polynomial time.
\newblock In {\em 19th Annual Symposium on Foundations of Computer Science
  (sfcs 1978)}, pages 75--83. IEEE, 1978.

\bibitem{ambainis2002two}
Andris Ambainis and John Watrous.
\newblock Two-way finite automata with quantum and classical states.
\newblock {\em Theoretical Computer Science}, 287(1):299--311, 2002.

\bibitem{ambainis2015automata}
Andris Ambainis and Abuzer Yakary{\i}lmaz.
\newblock Automata and quantum computing.
\newblock {\em arXiv preprint arXiv:1507.01988}, 2015.

\bibitem{anisimov1971group}
Ao~V Anisimov.
\newblock Group languages.
\newblock {\em Cybernetics and Systems Analysis}, 7(4):594--601, 1971.

\bibitem{arute2019quantum}
Frank Arute, Kunal Arya, Ryan Babbush, Dave Bacon, Joseph~C Bardin, Rami
  Barends, Rupak Biswas, Sergio Boixo, Fernando~GSL Brandao, David~A Buell,
  et~al.
\newblock Quantum supremacy using a programmable superconducting processor.
\newblock {\em Nature}, 574(7779):505--510, 2019.

\bibitem{birget2002isoperimetric}
J-C Birget, A~Yu Ol'shanskii, Eliyahu Rips, and Mark~V Sapir.
\newblock Isoperimetric functions of groups and computational complexity of the
  word problem.
\newblock {\em Annals of Mathematics}, pages 467--518, 2002.

\bibitem{borodin1983parallel}
Allan Borodin, Stephen Cook, and Nicholas Pippenger.
\newblock Parallel computation for well-endowed rings and space-bounded
  probabilistic machines.
\newblock {\em Information and Control}, 58(1-3), 1983.

\bibitem{condon1998power}
Anne Condon, Lisa Hellerstein, Samuel Pottle, and Avi Wigderson.
\newblock On the power of finite automata with both nondeterministic and
  probabilistic states.
\newblock {\em SIAM Journal on Computing}, 27(3):739--762, 1998.

\bibitem{dunwoody1985accessibility}
Martin~J Dunwoody.
\newblock The accessibility of finitely presented groups.
\newblock {\em Inventiones mathematicae}, 81(3):449--457, 1985.

\bibitem{dwork1990time}
Cynthia Dwork and Larry Stockmeyer.
\newblock A time complexity gap for two-way probabilistic finite-state
  automata.
\newblock {\em SIAM Journal on Computing}, 19(6):1011--1023, 1990.

\bibitem{dwork1992finite}
Cynthia Dwork and Larry Stockmeyer.
\newblock Finite state verifiers {I}: The power of interaction.
\newblock {\em Journal of the ACM (JACM)}, 39(4):800--828, 1992.

\bibitem{fefferman2020eliminating}
Bill Fefferman and Zachary Remscrim.
\newblock Eliminating intermediate measurements in space-bounded quantum
  computation, 2020.
\newblock \href {http://arxiv.org/abs/2006.03530} {\path{arXiv:2006.03530}}.

\bibitem{freivalds1981probabilistic}
R{\=u}si{\c{n}}{\v{s}} Freivalds.
\newblock Probabilistic two-way machines.
\newblock In {\em International Symposium on Mathematical Foundations of
  Computer Science}, pages 33--45. Springer, 1981.

\bibitem{greenberg1986lower}
Albert~G Greenberg and Alan Weiss.
\newblock A lower bound for probabilistic algorithms for finite state machines.
\newblock {\em Journal of Computer and System Sciences}, 33(1):88--105, 1986.

\bibitem{gromov1981groups}
Michael Gromov.
\newblock Groups of polynomial growth and expanding maps (with an appendix by
  {J}acques {T}its).
\newblock {\em Publications Math{\'e}matiques de l'IH{\'E}S}, 53:53--78, 1981.

\bibitem{grover1996fast}
Lov~K Grover.
\newblock A fast quantum mechanical algorithm for database search.
\newblock {\em Proceedings of the Twenty-Eighth Annual ACM Symposium of Theory
  of Computing}, pages 212--219, 1996.

\bibitem{harrow2009quantum}
Aram~W Harrow, Avinatan Hassidim, and Seth Lloyd.
\newblock Quantum algorithm for linear systems of equations.
\newblock {\em Physical review letters}, 103(15):150502, 2009.

\bibitem{hennie1965one}
Fred~C Hennie.
\newblock One-tape, off-line {T}uring machine computations.
\newblock {\em Information and Control}, 8(6):553--578, 1965.

\bibitem{herbst1991subclass}
Thomas Herbst.
\newblock On a subclass of context-free groups.
\newblock {\em RAIRO-Theoretical Informatics and Applications-Informatique
  Th{\'e}orique et Applications}, 25(3):255--272, 1991.

\bibitem{holt2005groups}
Derek~F Holt, Sarah Rees, Claas~E R{\"o}ver, and Richard~M Thomas.
\newblock Groups with context-free co-word problem.
\newblock {\em Journal of the London Mathematical Society}, 71(3):643--657,
  2005.

\bibitem{ibarra1988sublogarithmic}
Oscar~H Ibarra and Bala Ravikumar.
\newblock Sublogarithmic-space {T}uring machines, nonuniform space complexity,
  and closure properties.
\newblock {\em Mathematical systems theory}, 21(1):1--17, 1988.

\bibitem{kacneps1990minimal}
J{\=a}nis Ka{\c{n}}eps and R{\=u}si{\c{n}}{\v{s}} Freivalds.
\newblock Minimal nontrivial space complexity of probabilistic one-way {T}uring
  machines.
\newblock In {\em International Symposium on Mathematical Foundations of
  Computer Science}, pages 355--361. Springer, 1990.

\bibitem{karp1967some}
Richard~M Karp.
\newblock Some bounds on the storage requirements of sequential machines and
  {T}uring machines.
\newblock {\em Journal of the ACM (JACM)}, 14(3):478--489, 1967.

\bibitem{lipton1977word}
Richard~J Lipton and Yechezkel Zalcstein.
\newblock Word problems solvable in logspace.
\newblock {\em Journal of the ACM (JACM)}, 24(3):522--526, 1977.

\bibitem{loh2017geometric}
Clara L{\"o}h.
\newblock {\em Geometric group theory}.
\newblock Springer, 2017.

\bibitem{melkebeek2012time}
Dieter~van Melkebeek and Thomas Watson.
\newblock Time-space efficient simulations of quantum computations.
\newblock {\em Theory of Computing}, 8(1):1--51, 2012.

\bibitem{meyer1971economy}
Albert~R Meyer and Michael~J Fischer.
\newblock Economy of description by automata, grammars, and formal systems.
\newblock In {\em 12th Annual Symposium on Switching and Automata Theory (swat
  1971)}, pages 188--191. IEEE, 1971.

\bibitem{muller1983groups}
David~E Muller and Paul~E Schupp.
\newblock Groups, the theory of ends, and context-free languages.
\newblock {\em Journal of Computer and System Sciences}, 26(3):295--310, 1983.

\bibitem{nielsen2002quantum}
Michael~A Nielsen and Isaac Chuang.
\newblock Quantum computation and quantum information, 2002.

\bibitem{rabin1959finite}
Michael~O Rabin and Dana Scott.
\newblock Finite automata and their decision problems.
\newblock {\em IBM journal of research and development}, 3(2):114--125, 1959.

\bibitem{remscrim2019power}
Zachary Remscrim.
\newblock {The Power of a Single Qubit: Two-Way Quantum Finite Automata and the
  Word Problem}.
\newblock In {\em 47th International Colloquium on Automata, Languages, and
  Programming (ICALP 2020)}, volume 168 of {\em Leibniz International
  Proceedings in Informatics (LIPIcs)}, pages 139:1--139:18, 2020.

\bibitem{say2017magic}
AC~Say and Abuzer Yakaryilmaz.
\newblock Magic coins are useful for small-space quantum machines.
\newblock {\em Quantum Information \& Computation}, 17(11-12):1027--1043, 2017.

\bibitem{shallit1996automaticityIV}
Jeffrey Shallit.
\newblock Automaticity {IV}: sequences, sets, and diversity.
\newblock {\em Journal de th{\'e}orie des nombres de Bordeaux}, 8(2):347--367,
  1996.

\bibitem{shallit1996automaticityI}
Jeffrey Shallit and Yuri Breitbart.
\newblock Automaticity {I}: Properties of a measure of descriptional
  complexity.
\newblock {\em Journal of Computer and System Sciences}, 53(1):10--25, 1996.

\bibitem{shalom2010finitary}
Yehuda Shalom and Terence Tao.
\newblock A finitary version of gromov’s polynomial growth theorem.
\newblock {\em Geometric and Functional Analysis}, 20(6):1502--1547, 2010.

\bibitem{shepherdson1959reduction}
John~C Shepherdson.
\newblock The reduction of two-way automata to one-way automata.
\newblock {\em IBM Journal of Research and Development}, 3(2):198--200, 1959.

\bibitem{shor1994algorithms}
Peter~W Shor.
\newblock Algorithms for quantum computation: Discrete logarithms and
  factoring.
\newblock In {\em Proceedings 35th annual symposium on foundations of computer
  science}. Ieee, 1994.

\bibitem{ta2013inverting}
Amnon Ta-Shma.
\newblock Inverting well conditioned matrices in quantum logspace.
\newblock In {\em Proceedings of the forty-fifth annual ACM symposium on Theory
  of computing}, pages 881--890, 2013.

\bibitem{thom2013convergent}
Andreas Thom.
\newblock Convergent sequences in discrete groups.
\newblock {\em Canadian Mathematical Bulletin}, 56(2):424--433, 2013.

\bibitem{tits1972free}
Jacques Tits.
\newblock Free subgroups in linear groups.
\newblock {\em Journal of Algebra}, 20(2):250--270, 1972.

\bibitem{watrous2003complexity}
John Watrous.
\newblock On the complexity of simulating space-bounded quantum computations.
\newblock {\em Computational Complexity}, 12(1-2):48--84, 2003.

\bibitem{watrous2009encyclopedia}
John Watrous.
\newblock Encyclopedia of complexity and system science, chapter quantum
  computational complexity, 2009.

\bibitem{watrous2018theory}
John Watrous.
\newblock {\em The theory of quantum information}.
\newblock Cambridge University Press, 2018.

\bibitem{wolf1968growth}
Joseph~A Wolf et~al.
\newblock Growth of finitely generated solvable groups and curvature of
  riemannian manifolds.
\newblock {\em Journal of differential Geometry}, 2(4):421--446, 1968.

\bibitem{yakaryilmaz2010succinctness}
Abuzer Yakaryilmaz and AC~Cem Say.
\newblock Succinctness of two-way probabilistic and quantum finite automata.
\newblock {\em Discrete Mathematics and Theoretical Computer Science},
  12(4):19--40, 2010.

\end{thebibliography}

\end{document}